p

\documentclass{sig-alternate}

\usepackage{algorithm}
\usepackage{algorithmic}
\usepackage{graphicx}
\usepackage{epstopdf}
\usepackage[hang]{subfigure}
\usepackage{color}

\newtheorem{theorem}{Theorem}
\newtheorem{corollary}{Corollary}
\newtheorem{problem}{Problem}

\begin{document}

\conferenceinfo{Technical Report}{ArXiv}

\title{Spot: An accurate and efficient multi-entity device-free WLAN localization system
}
\numberofauthors{2}

\author{
\alignauthor
Ibrahim Sabek\\
       \affaddr{Comp. and Sys. Eng. Dept.}\\
       \affaddr{Alexandria University, Egypt}\\
       \email{ibrahim.sabek@alexu.edu.eg}
\and
\alignauthor
Moustafa Youssef\\
       \affaddr{Wireless Research Center}\\
       \affaddr{E-JUST, Egypt}\\
       \email{moustafa.youssef@ejust.edu.eg}
}

\maketitle

\begin{abstract}
Device-free (DF) localization in WLANs has been introduced as a value-added service that allows tracking indoor entities that do not carry any devices. Previous work in DF WLAN localization focused on the tracking of a single entity due to the intractability of the multi-entity tracking problem whose complexity grows exponentially with the number of humans being tracked.

In this paper, we introduce \textit{Spot} as an accurate and efficient system for multi-entity DF detection and tracking. \textit{Spot} is based on a probabilistic energy minimization framework that combines a conditional random field with a Markov model to capture the temporal and spatial relations between the entities' poses. A novel cross-calibration technique is introduced to reduce the calibration overhead of multiple entities to linear, regardless of the number of humans being tracked.  This also helps in increasing the system accuracy.

We design the energy minimization function with the goal of being efficiently solved in mind. We show that the designed function can be mapped to a binary graph-cut problem whose solution has a linear complexity on average and a third order polynomial in the worst case.
We further employ clustering on the estimated location candidates as a means for reducing outliers and obtaining more accurate tracking in the continuous space.
Experimental evaluation in two typical testbeds, with a side-by-side comparison with the state-of-the-art, shows that \textit{Spot} can achieve a multi-entity tracking accuracy of less than 1.1m. This corresponds to at least 36\% enhancement in median distance error over the state-of-the-art DF localization systems, which can only track a single entity. In addition, \textit{Spot} can estimate the number of entities correctly to within one difference error. This highlights that \textit{Spot} achieves its goals of having an accurate and efficient software-only DF tracking solution of multiple entities in indoor environments.

\end{abstract}

%


\newpage
\section{Introduction}

Device-free (DF) localization \cite{Youssef:DFPchallenges} is a concept that allows the detection and tracking of entities that do not carry any devices nor participate actively in the localization process. DF localization has a number of applications including intrusion detection, border protection, smart homes, and traffic estimation.

Different approaches have been proposed for addressing the DF detection and tracking problem that can be categorized into two main groups:  Those that require special hardware and those that leverage the already installed wireless infrastructure.

Radar-based systems, e.g. \cite{Yang:UWB,Lin:Doppler,Haimovich:MIMO}, computer vision systems, e.g. \cite{Moeslund:Camera,Krumm:Camera}, and radio tomographic imaging (RTI) \cite{Patwari:RTI} provide accurate detection and tracking. However, all require the installment of special hardware to track a DF entity. On the other  hand, systems that leverage the currently installed wireless networks, e.g. WLAN \cite{Youssef:DFPchallenges,Moussa:smart,Yang:Performin,Kosba:RASID,Seifeldin:Nuzzer}, provide a software only solution for DF localization and have the advantage of scalability in terms of cost and coverage area.

\begin{figure}[!t]
\centering
\includegraphics[width=3in]{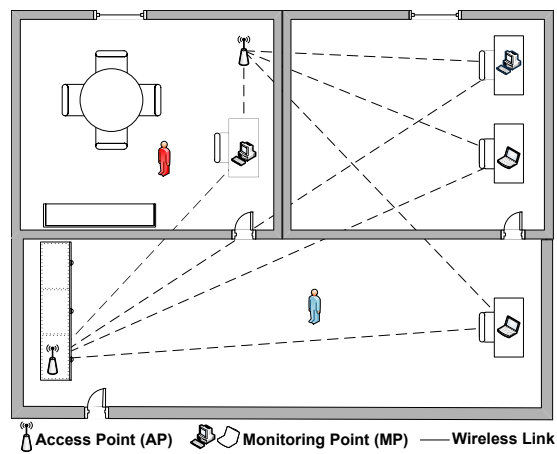}
\caption{Typical architecture of a DF WLAN localization system.}
\label{fig:DFP_arch}
\end{figure}

WLAN DF localization systems are based on the concept \cite{Youssef:DFPchallenges} that the presence of an entity in an RF environment affects the signal strength, which can be used to detect, track, and identify the entities. Figure~\ref{fig:DFP_arch} shows the architecture of a typical WLAN DF localization system. The system consists of signal transmitters (e.g. standard APs); signal receivers or monitoring points (MPs), such as any WiFi enabled device (e.g. laptops and APs themselves); and an application server that collects the received signal strength (RSS) for the different streams (where a stream is a single (AP, MP) pair) readings and processes them to detect events.

To track entities, and due to the complex relation between RSS and distance in indoor environment, a fingerprint has been traditionally used to capture the RSS behavior at different locations in the area of interest. To construct the fingerprint, a human stands at different locations in the area of interest and her effect on the RSS of the different streams is recorded at the MPs. Constructing the fingerprint for multiple entities though requires trying all the combinations of entities over all calibration locations, which grows exponentially with the number of fingerprint locations. Therefore, current effort in WLAN DF localization focuses only on the tracking of a \textbf{\emph{single entity}}.

In this paper, we introduce \textit{Spot} as a system for the accurate and efficient detection and tracking of \emph{\textbf{multiple DF entities}} in a WLAN environment. \textit{Spot} is based on a probabilistic energy minimization framework that combines a conditional random field with a Markov model:  Given a RSS vector of all the streams in the area of interest, the problem of estimating the most probable active user locations is mapped to an energy minimization problem whose potential function is designed to preserve \textit{smooth} and \textit{consistent} labels for active locations relative to their neighbors and their movement history. In addition, we show that the designed energy function is \emph{regular} in the sense that it can be mapped to a binary graph-cut problem whose solution has a linear complexity on average and a cubic polynomial in the worst case.
\textit{Spot} also introduces a novel cross-calibration technique to reduce the calibration overhead of multiple entities to \textbf{\emph{linear}} in the number of locations, as compared to \emph{\textbf{exponential}} for the current state-of-the-art.  This also helps in increasing the system accuracy.

Since a human can affect more than one location in the area of interest, we further employ clustering on the estimated location candidates as a means for reducing outliers and obtaining more accurate tracking in the continuous space. Each detected cluster represents a human whose location in the center of mass of fingerprint locations inside the cluster.
Experimental evaluation in two typical testbeds, with a side-by-side comparison with the state-of-the-art, shows that \textit{Spot} can achieve a tracking accuracy of less than 1.1m. This corresponds to at least 36\% enhancement in median error over the state-of-the-art DF localization systems in the two testbeds, while enabling the tracking of multiple entities. In addition, \textit{Spot} can estimate the number of entities with 100\% accuracy to within one difference error. This accuracy advantage is obtained without scarifying computational power.

In summary, the contribution of this paper is four-fold: (1) We formulate the \textbf{\emph{multi-entity}} DfP problem as an energy-minimization framework that preserves both spatial and temporal smoothness and consistency (Section~\ref{sec:Spot}), (2) We show how to map the problem to a binary graph-cut problem and obtain its solution efficiently; and present our novel cross-calibration technique that reduces the calibration complexity to linear in the number of locations, rather than exponential as in the current state-of-the-art (Section~\ref{sec:prob}), (3) We present clustering techniques for reducing noise and enhancing accuracy (Section~\ref{sec:cluster}), (4) We evaluate the system in two typical WiFi testbeds and compare it to the state-of-the-art DF WLAN localization techniques (Section~\ref{sec:eval}).

 We also discuss issues related to the system design and present future directions in Section~\ref{sec:discuss}. Related work and paper conclusions are presented in sections \ref{sec:related} and \ref{sec:conclude} respectively.

\section{The Spot System}
\label{sec:Spot}
In this section, we give the details of \textit{Spot}. We start by an overview of the system architecture followed by the details of the system modules.

\begin{figure}[!t]
\centering
\includegraphics[width=3in]{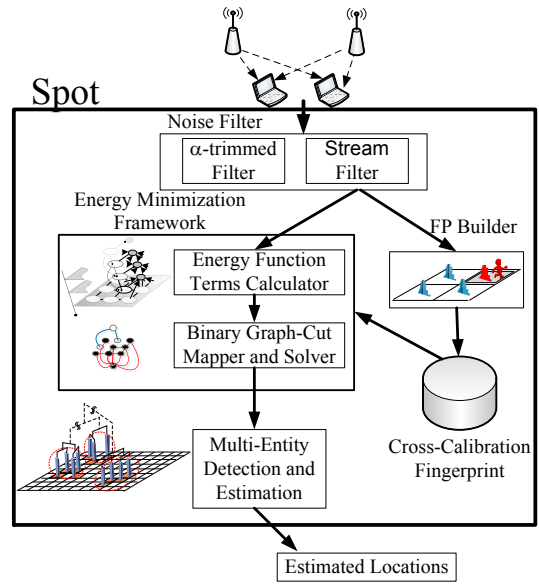}
\caption{Spot system architecture.}
\label{fig:arch}
\end{figure}

\subsection{Overview}
\label{sec:overview}
Figure~\ref{fig:arch} shows the system architecture. The system collects the signal strength readings from the monitoring points for processing. There are two phases of operation:

\begin{enumerate}
  \item Offline training phase: to estimate the system parameters based on the collected signal strength readings and construct the device-free fingerprint. During this phase, a human stands at different locations in the area of interest and the RSS at each MP is recorded.
      Note that our formulation requires only one human for calibration in the offline phase, regardless of the number of humans during the system operation (Section~\ref{sec:FP}). This significantly reduces the calibration overhead as compared to the state-of-the-art DF systems.
  \item Online tracking phase:  to estimate the multi-entities' locations based on
the received signal strength from each stream and the fingerprint prepared in the offline phase using the energy minimization framework.
\end{enumerate}

The \emph{Noise Filtering} module reduces the noise in the RSS readings and filters outlier streams.

The \emph{Energy Minimization Framework}  calculates the probabilities used in the energy minimization framework, constructs an equivalent graph, and estimates the most probable active locations (i.e. environment map) based on solving a binary graph-cut problem.

The \emph{Multi-Entity Detection and Estimation} module uses clustering techniques to estimate the number of entities and the location of each entity. A non-zero estimated number of entities is equivalent to a detection event in the area of interest.
\subsection{System Model}
\label{sec:model}
Without loss of generality, let $\mathbb{X}$ be a
2-dimensional physical space. At each location $x \in \mathbb{X}$,
we can get the signal strength from $k$ streams. We denote the
$k$-dimensional signal strength space as $\mathbb{S}$. Each element
in this space is a $k$-dimensional vector, $s= (s_1, ...,
s_k)$, whose entries represent the signal strength readings from a different (AP, MP) pair. We further assume
that the samples from \textit{different} streams are independent.


Given that $m$ humans are standing in the area of interest, $m \ge 0$, these humans will affect the different streams. 
Therefore, the problem becomes:

\begin{problem}
\label{pro:main}
  We want to both estimate the number of humans $\hat{m}$ and, if $\hat{m} > 0$, the locations of these humans $\{x_i| 0< i \le \hat{m} , x_i \in \mathbb{X}\}$, such that the probability $P(x_1, x_2, ..., x_{\hat{m}}|s)$ is maximized.
\end{problem}

In Section \ref{sec:prob}, we assume a discrete $\mathbb{X}$ space. We
discuss the continuous space case in Section ~\ref{sec:cluster}.

\subsection{Noise Filtering}
\label{sec:filter}

The aim of this module is to preprocess collected RSS readings during the offline and online phases to reduce the noise effects and detect outliers. We use two techniques: RSS filtering and stream filtering.
\subsubsection{RSS Filtering}
RSS is a noisy quantity due to the time varying wireless channel \cite{WCNC03}. To reduce the noise effect, we apply an $\alpha$\textit{-trimmed Mean} filter \cite{Yang:Performin} to the measured RSS values. An $\alpha$-trimmed filter has the advantage of handling both impulse and gaussian noise, as compared to mean and median filters that can handle only one of them. In addition, it is simple to implement: Given a window of $q$ RSS samples, the $\alpha$-trimmed filter sorts the samples (such that  $RSS_1 \leq RSS_2 \leq \cdots \leq RSS_q$) and then discards the $\alpha$ extreme samples and averages the remaining samples.
The output of the $\alpha$-trimmed mean filter is:
\begin{equation}
 f(q; \alpha) = \frac{1}{q- 2\left\lceil \alpha q \right\rceil} \sum_{i=\left\lceil \alpha q \right\rceil+1}^{q-\left\lceil \alpha q \right\rceil}{\textrm{RSS}_i}
\label{eq:15}
\end{equation}
 where $0 \leq \alpha < 0.5$. We set $\alpha$ to 0.2 as it is a reasonable value for the window size we use in our system (Section~\ref{sec:eval}).

\subsubsection{Streams Filter}

Even after smoothing the RSS values, using the alpha-trimmed filter, the readings of a single stream may have significantly changed between the offline and online phases due to changes in the environment. To detect this change and filter outlier streams, we use the Analysis of Variance (ANOVA)  to test whether the mean of the RSS of a particular stream have significantly changed between the offline and online phase. If there is a statistically significant difference, the stream is filtered from the current calculations.

\section{Energy Minimization Framework}
\label{sec:prob}

In this section, we assume a discrete $\mathbb{X}$ space with $n$ locations. Let $\{\alpha_i^t, 0<i < n$\} be a set of bernoulli random variables, where $\alpha_i^t$ takes the value of 1 if a human is standing at location $i \in \mathbb{X}$ at time $t$, and 0 otherwise.  Therefore, the problem of estimating the number of entities $\hat{m}$ and their locations, given the received signal strength vector $s$ (Problem \ref{pro:main}),
is equivalent to finding the assignment of $\alpha_i^t$'s that maximizes \begin{equation}
P(\mathbb{M}^t|s)
\label{eq:no_temporal}
\end{equation}

where $\mathbb{M}^t= (\alpha_1^t, \alpha_2^t, ..., \alpha_n^t)$. We refer to $\mathbb{M}^t$ as the \emph{\textbf{environment map}} at time $t$. In this case, $\hat{m}= \sum_{i=1}^{n} \alpha_i^t$ and the most probable locations of the $\hat{m}$ entities are the locations whose $\alpha_i^t$'s are assigned to one.

Traditional work on probabilistic WLAN localization, \textbf{\emph{both device-based and device-free}}, e.g. \cite{Youssef:Horus,Seifeldin:Nuzzer_Report}, use Bayesian inversion to estimate $P(\mathbb{M}^t|s)$. However, these systems typically assume only one entity in the area of interest. Moving to more than one entity makes this Bayesian inversion approach intractable as the complexity of estimating $P(s|\mathbb{M}^t)$ increases exponentially with the number of entities that need to be tracked \cite{Seifeldin:Nuzzer_Report} (due to the need to try all combinations of humans' poses in the area of interest).

Alternatively, we use an energy-minimization framework that leverages this joint estimation problem of  $\alpha_i^t$'s to enhance the accuracy while, at the same time, obtains an efficient solution.

In particular, we represent the spatial constraints on the human position by a Conditional Random Field (CRF) model favoring coherence between adjacent locations. The temporal relation between the human locations is captured by a second order Hidden Markov Model (Figure~\ref{fig:model}). Estimation is finally performed by mapping the problem to a binary graph-cut problem that can be efficiently solved in $O(n)$ on average and $O(n^3)$ in the worst case.

A CRF is an undirected graphical model that defines a log-linear distribution over label sequences given a particular observation sequence \cite{Lafferty:CRF}. It was introduced as a framework for labeling and segmenting data that models the conditional probability $P(Y|X)$, where $X$ and $Y$ are the observations and the labels respectively. CRFs have the advantage of relaxing the strong independence assumptions made by Hidden Markov Models \cite{Sutton06:CRFintroduction} for a large number of variables (such as those in the environment map). In addition, CRFs avoid the label bias problem [8], a weakness exhibited by maximum entropy Markov models [9] (MEMMs) and other conditional Markov
models based on directed graphical models. Therefore, CRFs outperform both MEMMs
and HMMs on a number of real-world sequence labeling tasks [8, 11, 15].

\begin{figure}[!t]
\centering
\includegraphics[width=3.5in]{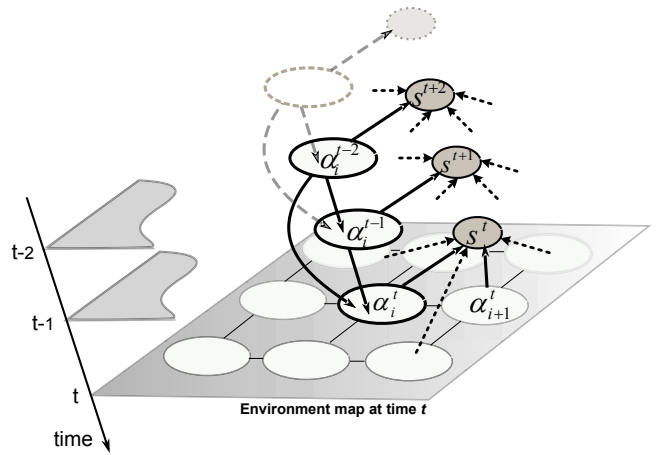}
\caption{\textbf{Combined CRF-HMM model}. This graphical
model illustrates both the signal strength likelihood
together with the spatial and temporal priors. The same
temporal chain is repeated at each discrete location. Spatial dependencies
are illustrated for a 4-neighborhood system. The entire environment map affects the RSS vector $s^t$.}
\label{fig:model}
\end{figure}
By combining a HMM for temporal relations with a CRF for spatial relations, we gain the benefit of both worlds in terms of accuracy and efficiency. In the rest of this section, we describe the construction of the energy minimization framework and how we efficiently solve it.

\subsection{Framework Construction}
Our model extends the model in Equation~\ref{eq:no_temporal} to capture the temporal constraints. In particular, our goal becomes to find the environment map at time $t$, $\mathbb{M}^t$, that maximizes:
\begin{equation}
P(\mathbb{M}^t|s^t,\mathbb{M}^{t-1}, \mathbb{M}^{t-2})
\label{eq:temporal}
\end{equation}

assuming a second order temporal dependence in the Markov model as we discuss in details later.

Based on CRF theory \cite{Kumar:2003}, our combined model estimates the probability of Equation~\ref{eq:temporal} as

\begin{equation}
  P(\mathbb{M}^t|s^t,\mathbb{M}^{t-1}, \mathbb{M}^{t-2})  \propto \exp - \left\{ E^{i} \right\}
\label{eq:eng}
\end{equation}

where $E^{i}= E(s^t,\mathbb{M}^{t}, \mathbb{M}^{t-1}, \mathbb{M}^{t-2})$ is an energy function that captures the required constraints on the DF tracking problem. That is, we want to estimate the current environments map given the previous two environment maps and the current signal strength vector measured at the monitoring points. This is obtained by the joint maximization of the posterior in Equation \ref{eq:eng}, which is equivalent to the minimization of energy:

\begin{equation}
  \hat{\mathbb{M}}^t = (\hat{\alpha}_1^t, \hat{\alpha}_2^t, ..., \hat{\alpha}_n^t) = \arg \min E^{i}
\label{eq:eng_min}
\end{equation}

\textbf{Energy Terms:} For our DF tracking problem, each $E^i$ is composed of three components:
\begin{equation}
\begin{split}
  E^{i} &= E(s^t,\mathbb{M}^{t}, \mathbb{M}^{t-1}, \mathbb{M}^{t-2})\\
  &= V^{\textrm{Tm}}(\mathbb{M}^{t}, \mathbb{M}^{t-1}\mathbb{M}^{t-2})+ V^{\textrm{Sp}}(\mathbb{M}^{t}, s^t)+ U^{\textrm{SS}}(\mathbb{M}^{t}, s^t)
\end{split}
\label{eq:e_i}
\end{equation}

The term $V^{\textrm{Tm}}(\mathbb{M}^{t}, \mathbb{M}^{t-1}\mathbb{M}^{t-2})$ is a temporal prior term that represents a second-order Markov chain that imposes a tendency to temporal continuity on the environment map.

The term $V^{\textrm{Sp}}(\mathbb{M}^{t}, s^t)$ presents a spatial prior term which imposes a
tendency to spatial continuity of the environment map, favoring coherent assignments.

Finally, the $U^{\textrm{SS}}(\mathbb{M}^{t}, s^t)$ term is a likelihood term that evaluates the evidence
for location labels based on the RSS distributions in the case of human absence and presence.

This energy model captures both the signal strength likelihood together with the spatial and temporal priors. Figure \ref{fig:model} shows the graphical representation of the model. Details of these factors are given in the next subsections.

\begin{figure}[!t]
\centering
\includegraphics[width=3in]{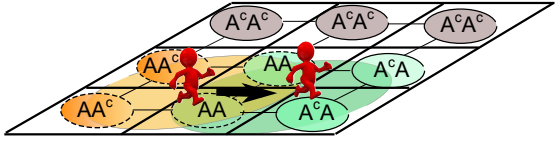}
\caption{Temporal transitions at a location (assuming the human affects four locations inside the circle). (a) An entity moves towards the right from time $t-2$ to time $t-1$. (b) Between the
two time instances, locations may remain in their own active or inactive states (denoted $A$ and $A^c$, respectively) or change state; thus
defining four different kinds of temporal transitions: $A \rightarrow A$,
$A \rightarrow A^c$, $A^c \rightarrow A$, $A^c \rightarrow A^c$. Those transitions influence the label
that a location in the environment map is going to assume at time $t$.}
\label{fig:temporal}
\end{figure}

\newpage
\subsubsection{Temporal prior term}
Figure \ref{fig:temporal} shows the four different temporal
transitions a location assignment (label) can undergo in an environment map, based on
a two time instances analysis. For instance, an active location may
remain active (locations labeled $AA$ in Figure~\ref{fig:temporal}) or
move to the inactive state (locations labeled $A A^c$) etc. It is important to note that a first-order Markov chain is inadequate
to convey the nature of temporal coherence in this
problem; a higher order Markov chain is required. For example,
a location that was inactive at time $t - 2$
and is active at time $t - 1$ is far more likely to
remain active at time $t$ than to go back to the
inactive state. A second-order Markov chain is used to balance performance and complexity. We quantify the effect of the order of the chain in Section~\ref{sec:res_hmm}.

These intuitions are captured probabilistically and incorporated
in our energy minimization framework by means of
a second order Markov chain, as shown in the graphical
model of Figure~\ref{fig:model}. The temporal transition priors ($P(\alpha_i^t|\alpha_i^{t-1}, \alpha_i^{t-2})$) are
learned during the training phase. This leads to the following joint temporal prior term:
\begin{equation}
V^{\textrm{Tm}}(\mathbb{M}^{t}, \mathbb{M}^{t-1}\mathbb{M}^{t-2})= \beta \sum_{i=1}^n -[\log P(\alpha_i^t|\alpha_i^{t-1}, \alpha_i^{t-2})]
\label{eq:v_t}
\end{equation}
where  $\beta < 1$ is a discount factor to allow for multiple
counting across non-independent locations. The optimal value of $\beta$, as well as the other parameters of
the CRF, are obtained discriminatively from the training data using the iterative scaling algorithm \cite{Lafferty:2001}.

\subsubsection{Spatial prior term}
This term should favor coherent environment maps, i.e. adjacent locations have similar labels. We adapt a variation of the Ising model commonly used for segmentation applications \cite{Boykov01:boundary} where the spatial energy term can be represented as:

\begin{align}
	V^{\textrm{Sp}}(\mathbb{M}^{t}, s^t)
	= \sum_{\{c_i,c_j\} \in \mathbb{N}} V_{\{c_i,c_j\}}^{\textrm{Sp}} (\alpha_{c_i}^{t},\alpha_{c_j}^{t},s^t) & \nonumber \\
	= \gamma \sum_{\{c_i,c_j\} \in \mathbb{N}, \alpha_{c_i}^{t} \neq \alpha_{c_j}^{t}}  \left(\frac{1 + e^{-\left\|P(s^t|\alpha_{c_i}^{t})-P(s^t|\alpha_{c_j}^{t})\right\|^2}}{2}\right)
\label{eq:spatial_prior_term}
\end{align}
where 
$\mathbb{N}$ is the set of pairs of neighboring locations. The term $P(s^t|\alpha_{c_i}^{t})$ represents the conditional probability of receiving the signal strength vector $s^t$ when the human is present at location $c_i$ ($\alpha_{c_i}^{t}=1$) or not present ($\alpha_{c_i}^{t}=0$). This can be estimated during the training phase as described in Section~\ref{sec:FP}.
The constant $\gamma$ is a strength parameter for the coherence
prior that can be estimated based on the training data.

\subsubsection{Likelihood for signal strength}
The term $U^{\textrm{SS}}(\mathbb{M}^{t}, s^t)$ is the log likelihood of the received signal strength. The term is defined as :
\begin{align}
	U^{\textrm{SS}}(\mathbb{M}^{t}, s^t)
	&= \delta \sum_{i=1}^n \left[ - \log P (s^t~|~ \alpha_i^{t})\right]
\label{eq:rss_likelihood_term}
\end{align}

where  $\delta < 1$ is a discount factor to allow for multiple
counting across non-independent locations whose optimal value is obtained discriminatively from the training data.

RSS likelihoods are learned during the offline training phase as described in the next section.

\begin{figure}[!t]
\centering
\subfigure[Traditional FP construction for one entity]{\includegraphics[width=3.5in]{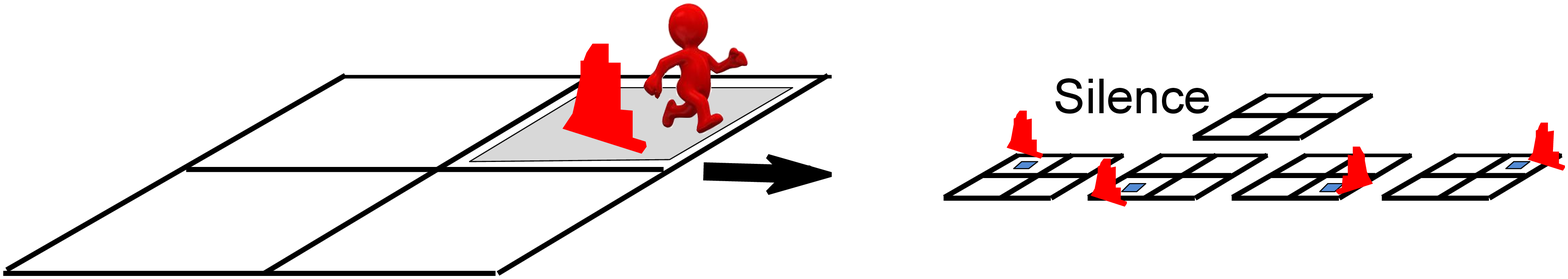}}
\subfigure[Traditional FP construction for multi-entities]{\includegraphics[width=3.5in]{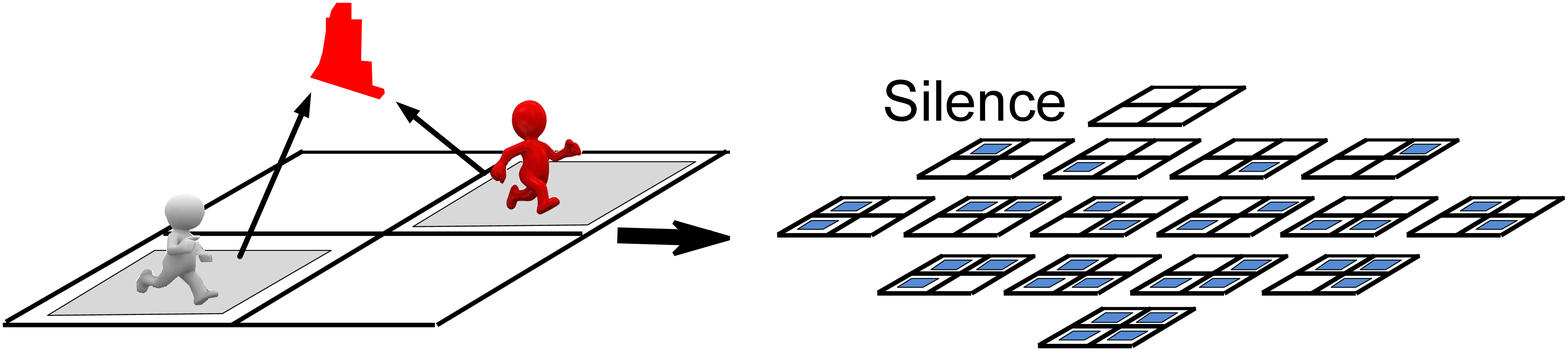}}
\subfigure[FP construction for Spot]{\includegraphics[width=3.5in]{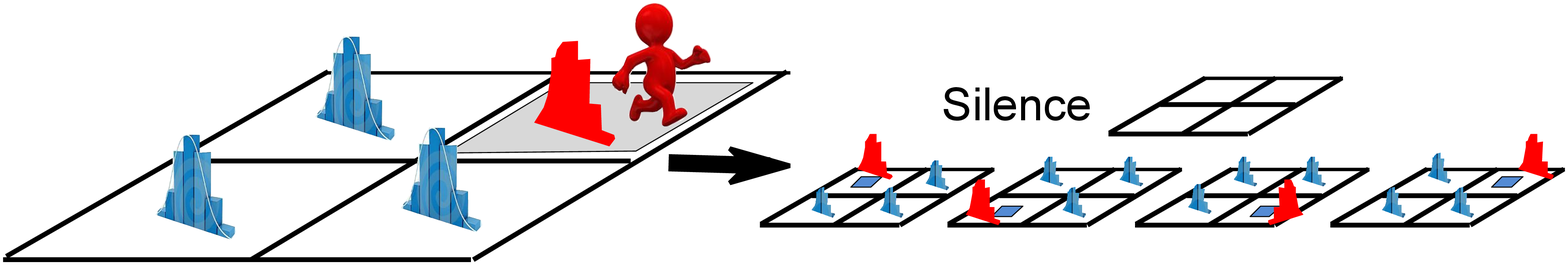}}
\caption{\textbf{Difference between fingerprint (FP) construction for traditional DF systems and Spot}. Left figure represents an example while the figure to the right represents all required combinations (FP complexity). (a)  FP construction for one entity in a traditional DF system: One histogram, representing the active state RSS, is stored in only one location (where the user is standing). (b) FP construction for two entities in a traditional DF system: Two humans are needed along with trying all their poses combinations in the area of interest ($\binom{n}{2}$). A total of $2^n$ combinations are required to capture the fingerprint of all possible number of humans and their locations. (c) FP construction in Spot: Only one human is needed to construct the FP regardless of the actual number of humans to be tracked (due to the environment map formulation). A human standing at one location ($x$) captures the RSS active histogram at this location ($P(s^t~|~ \alpha_x^{t}=1$)) and affects the inactive histograms at all other FP locations, ($P(s^t~|~ \alpha_i^{t}=0, \forall i \neq x$)), (cross-calibration). This leads to two histograms at every FP location.}
\label{fig:FP}
\end{figure}

\subsection{Fingerprint Construction}
\label{sec:FP}
During the offline phase, \textit{Spot} needs to estimate both the RSS likelihood, $P(s^t~|~ \alpha_i^{t})$, and the temporal transition priors, $P(\alpha_i^t|\alpha_i^{t-1}, \alpha_i^{t-2})$. This is the functionality of the \emph{Fingerprint Builder Module}.

\subsubsection{RSS likelihood}
Based on the described signal strength terms in the energy function, i.e. the spatial prior and signal strength likelihood, the fingerprint of \textit{Spot} is \textbf{\emph{unique}} among all the previous device-based and device-free WLAN localization systems. Figure~\ref{fig:FP} shows the difference between the fingerprint for a traditional DF system and that of \textit{Spot}.  In particular, we use a \textbf{\emph{cross-calibration}} technique, where an entity standing at location $x$ contributes to the \textit{active} RSS likelihoods of $x$  ($P(s^t~|~ \alpha_x^{t}=1$)) and the \textit{inactive} RSS likelihoods of the all remaining $n-1$ FP locations ($P(s^t~|~ \alpha_i^{t}=0, \forall i \neq x$)). This has two advantages: (1)  It reduces the the coverage sparsity problem in the presence of few streams and (2) it converts the intractable exponential number of cases of building the fingerprint for traditional DF systems \cite{Seifeldin:Nuzzer_Report} to a linear complexity problem, as only one human is needed for training, regardless of the number of humans to be tracked.

In summary, at each location, we have two histograms for the RSS corresponding to the active and inactive states respectively using the cross-calibration technique. The fingerprint is the collection of these two histograms over all locations $x \in \mathbb{X}$. We smooth the generated histograms by convolution with separable gaussian kernels to avoid the zero-probability problem of missing values in the training set.

\begin{figure}[!t]
\centering
\includegraphics[width=2.5in]{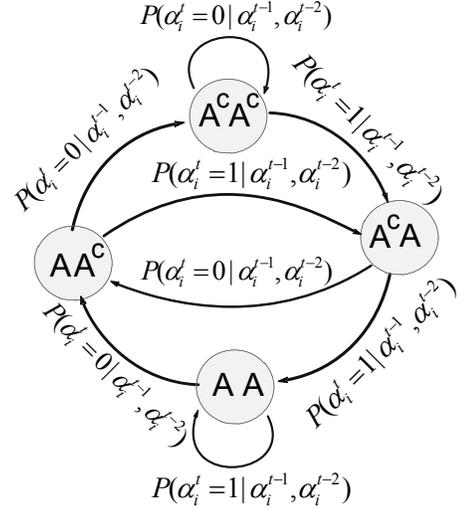}
\caption{Finite state diagram for the possible temporal transitions at any location. The sum of arcs originating from any node is one, leading to only four degrees of freedom.}
\label{fig:temporal_state_diagram}
\end{figure}

\newpage
\subsubsection{Temporal transition prior}
Although there are eight possible
transitions (Figure~\ref{fig:temporal_state_diagram}), due to probabilistic normalization ($P(\alpha_i^t =
1|\alpha_i^{t-1}, \alpha_i^{t-2}) = 1 - P(\alpha_i^t =
0|\alpha_i^{t-1}, \alpha_i^{t-2})$), the temporal
priors have only four degrees of freedom. These temporal transition priors are learned from the training data.

\subsection{Most Probable Map Estimation}
In this section, we show how to obtain the optimal environment map by solving the energy minimization problem in Equation~\ref{eq:eng_min} efficiently through mapping it to a binary graph-cut problem. We start by a brief background on graph-cuts, followed by how to map the DF energy minimization problem to a graph problem.

\subsubsection{Binary graph-cuts}
Let $\cal{G}= (\cal{E}, \cal{V})$ be a directed graph with nonnegative
edge weights that has two special vertices (terminals): the source $s$ and the sink $t$. An $s-t$-cut (or a binary graph-cut) $C= \{S, T\}$ is a partition of the vertices of $\cal{V}$ into two disjoint sets $S$ and $T$ such that $s \in S$ and $t \in T$. The cost of the cut is the sum of costs of all edges that go from $S$ to $T$:
\[
c(S, T)= \sum_{u \in S, v \in T, (u, v) \in \cal{E}} c(u, v)
\]
The minimum $s-t$-cut problem is to find a cut $C$ with the
smallest cost. Ford and Fulkerson \cite{ford1962flows} proved that this is equivalent to computing the maximum flow from the source to sink. This problem can be solved in a low order polynomial in $n$ \cite{Boykov:2001} \footnote{Note however that
generalizations of the minimum s-t-cut problem to
involve more than two terminals are NP-hard. We prove in the next subsection that our problem can be mapped to a binary graph-cut problem.}.
This way, a binary graph-cut can be considered as a binary labeling of the graph nodes to be either $s$ or $t$.

\begin{figure}[!t]
\centering
\hspace{-.8in}
\includegraphics[width=2.5in]{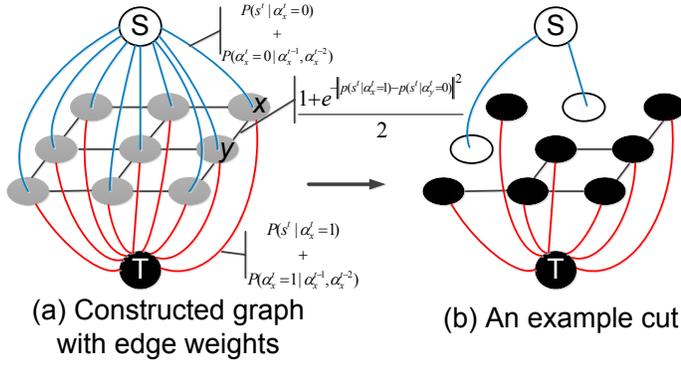}
\caption{Mapping the DF energy minimization problem to a binary graph-cut problem.}
\label{fig:graph}
\end{figure}

\subsubsection{DF tracking as a binary graph-cut problem}
\label{sec:DF_binary_graph_cut}
Not every energy minimization function can be solved using a graph-cut approach. According to \cite{Kolmogorov:2002}, the following theorem  gives a necessary and sufficient condition for a function to be solved using the binary min-cut algorithm.
\begin{theorem}
  Let $E$ be a function of $n$ binary variables in the form of
  \[
  E(x_1, ..., x_n)= \sum_i E^i(x_i)+ \sum_{i<j} E^{ij} (x_i,x_j)
  \]
  Then, $E$ is graph-representable if, and only if, each term $E^{ij}$ satisfies the inequality
  \begin{equation}
      E^{ij} (0,0)+   E^{ij} (1,1) \leq   E^{ij} (1,0)+   E^{ij} (0,1).
      \label{eq:graph}
  \end{equation}
\end{theorem}

Note that the condition only involves the binary terms, i.e. those that involve the relation between two variables. This maps only to the spatial consistency term in our DF energy function (Equation~\ref{eq:spatial_prior_term}).
\begin{corollary}
  The DF energy minimization function is graph-representable.
\end{corollary}
\begin{proof}
  The proof follows directly by mapping the terms of  Equation~\ref{eq:spatial_prior_term} to Equation~\ref{eq:graph} noting that the LHS of Equation~\ref{eq:graph} is zero in the DF tracking problem and the two RHS terms are positive.
\end{proof}

 The above corollary tells us that we can find a polynomial time efficient solution to the DF energy minimization problem using the binary graph-cut mapping. Figure~\ref{fig:graph} shows how our energy minimization problem can be mapped to a binary graph-cut problem. We construct a graph that has $n+ 2$ nodes, where $n$ nodes are the original discrete environment map locations and two additional nodes are added to represent the source and sink nodes. There are two types of edges. Those between the original discrete environment map locations (n-edges) and those between each node and the source and sink terminal nodes (t-edges). The edge weights are assigned in the following way to guarantee that the min-cut solution to this graph is equivalent to minimizing the energy function in Equation~\ref{eq:eng_min} \cite{Kolmogorov:2002}:
\begin{enumerate}
  \item The t-edge between the source and node $x$ is assigned a weight of $P(s^t | \alpha_x^t=0)+ P(\alpha_x^t=0|\alpha_x^{t-1}, \alpha_x^{t-2})$.
  \item The t-edge between node $x$ and the sink is assigned a weight of $P(s^t | \alpha_x^t=1)+       P(\alpha_x^t=1|\alpha_x^{t-1}, \alpha_x^{t-2})$.
  \item The n-edge $(x, y)$ between node $x$ and node $y$ is assigned a weight of $\frac{1 + e^{-\left\|P(s^t|\alpha_{x}^{t}=1)-P(s^t|\alpha_{y}^{t}=0)\right\|^2}}{2}$.
\end{enumerate}

\begin{theorem}
  The binary graph-cut solution on the constructed graph is a solution to the corresponding energy minimization problem in Equation \ref{eq:eng_min}.
\end{theorem}

\begin{proof}
The proof is in the appendix.
\end{proof}

Any node connected to the source (sink) node after the cut is considered inactive (active).

\subsection{Computational Complexity}
The binary graph-cut algorithm requires $O(n^3)$ operations, where $n$ is the number of fingerprint locations. However, we use the algorithm in \cite{Boykov:2001} as it provides an iterative fast algorithm. Although the algorithm has the same complexity in the worst case, its average complexity is $O(n)$. This has been confirmed in our experiments.

\subsection{Discussion}
Using the proposed technique, we could reduce the training  complexity from $O(2^n)$ to $O(n)$. This is a significant reduction in the calibration overhead which turns the multi-entity tracking problem to a feasible problem.

The proposed framework also treats the detection and tracking problem in a homogenous manner. In particular, detection can be regarded as a special case of the system, where a non-zero estimate of the number of entities  is equivalent to a detection event.

\section{Multi-Entity Detection and Estimation}
\label{sec:cluster}
 The output of the binary graph-cut operation is a set of candidate locations. However, these locations cannot be used directly as a human presence at a location typically affects the signal strength at \emph{more than one} neighboring location (Figure \ref{fig:temporal}) leading to overestimating the actual number of humans and their locations. This effect on neighboring locations decreases as we move away from the actual human location. Therefore, the \textit{Multi-entity Detection and Estimation Module} applies clustering to the output of the binary graph-cut algorithm, such that the number of output clusters determines the number of entities and the center of mass of each cluster gives the coordinates of human corresponding to this cluster. This not only solves the problem of overestimating the number of entities, but also in locating the entities in the continuous space by the weighted averaging of all the samples within a cluster.
  To further enhance accuracy, we apply clustering to the last $w$ environment maps by merging them into one map.

\subsection{Approach}
We used a hierarchical clustering technique (Figure \ref{fig:cluster}) as it gives us an intuitive means to estimate the number of clusters. In particular, leaf nodes represent individual candidates. Each internal node represents a possible cluster. As we go up in the tree, clusters are combined to form a bigger cluster using Euclidean distance between clusters centers as a similarity measure. The root of the tree corresponds to one cluster that contains the entire set of candidate nodes.
Starting from the root of the tree, if the degree of inconsistency between two clusters is high, based on a parameter $r$, we split them as two separate clusters. This process is repeated recursively for each of the split clusters until the degree of inconsistency is below $r$. The final number of clusters represents the estimated number of humans and the center of mass of each cluster is the estimated human location.

\subsection{Clustering Complexity}
The hierarchical clustering requires $O(c^3)$ operations, where $c$ is the number of candidate locations. Typically, $c$ is $<<n$. Therefore, clustering has a low overhead. We quantify this effect in Section~\ref{sec:eval}.

\begin{figure}
\centering
\includegraphics[width=3.0 in]{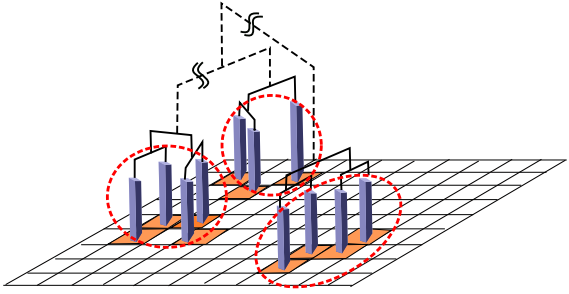}
\caption{An example dendrogram showing the hierarchical clustering inconsistency. Dotted lines represent split clusters. The figure shows three clusters corresponding to three entities in the area of interest.}
\label{fig:cluster}
\end{figure}

\section{Performance Evaluation}
\label{sec:eval}

In this section, we analyze the performance of \textit{Spot} and compare it to a deterministic \cite{Seifeldin:Nuzzer} and a probabilistic \cite{Seifeldin:Nuzzer_Report} state-of-the-art DF WLAN localization systems. We start by  describing the experimental setup and data collection. Then, we analyze the effect of different parameters on the system performance.

\subsection{Testbeds and Data Collection}
We evaluate our system in two different testbeds (Figure~\ref{fig:testbeds}). The first testbed covers a residential apartment with an area of $114 m^2$ (about 1228 sq. ft.) while the second testbed represents an office building with an area of $130 m^2$ (about 1400 sq. ft.). The two testbeds were covered by TP-link TL-WA500G APs and D-Link Airplus G+DWL-650 wireless NICs.

For data collection, we used a sampling rate of one hertz. We had six RSS data streams for both testbeds. A total of 25 fingerprint locations, uniformally distributed over the testbed, are sampled for both testbeds. An \textbf{\emph{independent test set}} at 17 (22) test locations for the first (second) testbed, was collected at different times and by different persons.

We give the details of the results of the first testbed and summarize the results of the second. Figure 10
shows an example of the output of the system.

\begin{figure}[!t]
\centering
\subfigure[Testbed 1]{\includegraphics[width=1.55in]{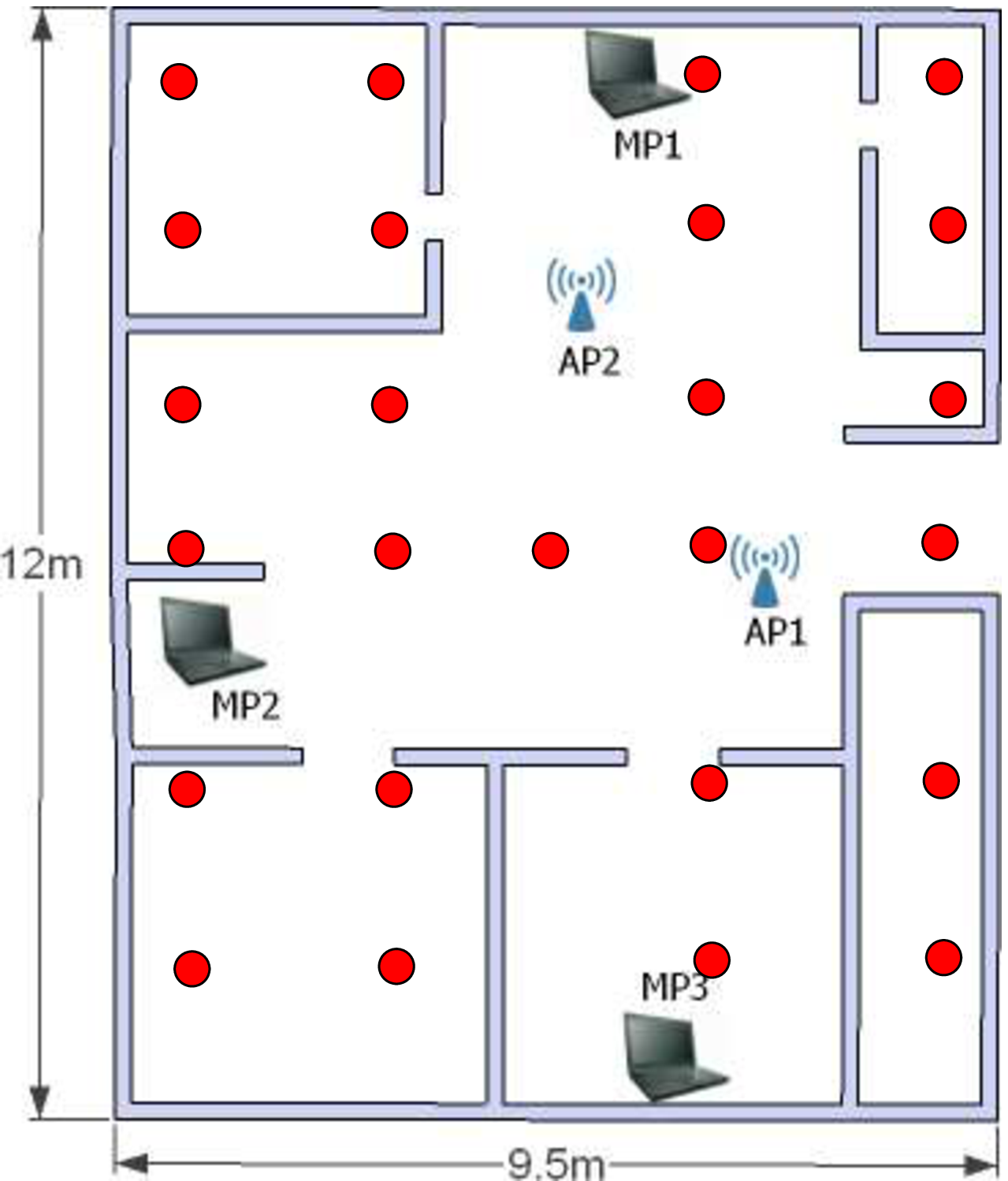}
\label{fig:homefloor_layout}
}
\subfigure[Testbed 2]{
\includegraphics[width=1.55in]{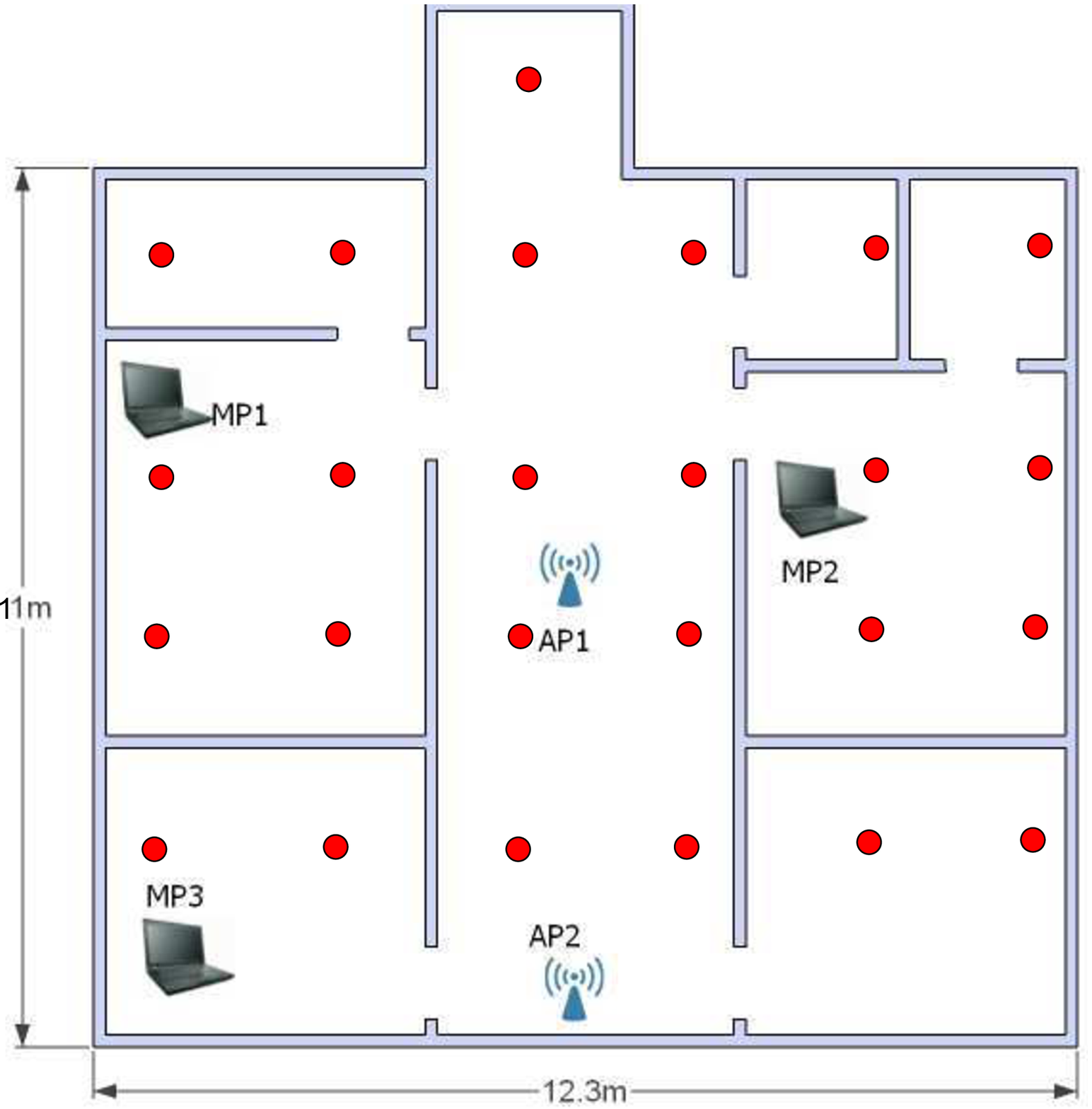}
\label{fig:villa_layout}
}
\caption{Experimental testbeds.}
\label{fig:testbeds}
\end{figure}

\begin{table}
    \centering
    \caption{Default parameters values.}
    \begin{tabular}{|c|c|l|} \hline
     Parameter & Default value & Meaning\\ \hline \hline
        $k$ & 6  & Num. of used streams\\ \hline
        $n$ & 25  & Num. of FP locations\\ \hline
        $w$ & 13  & Window size\\ \hline
        $o$ & 2  & HMM order\\ \hline
        $r$ & 0.25  & Clust. inconsistency thr.\\ \hline
    \end{tabular}
    \label{tab:default}
\end{table}
\subsection{Parameters Effect}
 In this section, we study the effect of changing the system parameters on the performance of \textit{Spot}. The average distance error is used as the main metric where the error is calculated as the difference between the estimated location and the closest ground truth location (for multiple-entities). We present two versions of the average distance error relying on the level of details needed. If determining the zone the person is standing in the main target, then the average distance error is calculated based on the centers of estimated and actual zones and we call it {\bf zones-based difference} (Figures 11-16). On the other hand, if higher level of details is required, i.e. the exact person location, we calculate the difference based on the coordinates of the ground truth and the estimated location and we call it {\bf locations-based difference} (Figures 19-24).

To calculate the distance error for multiple entities, we use the Euclidean distance between the estimated zone/location of each entity and the closest fingerprint zone/location\footnote{If the estimated number of entities is less than the actual number of entities, we use the testbed center as the ground truth.}.

 Table~\ref{tab:default} shows the default values of the different parameters.

\subsubsection{Window size ($w$)}

Figures~11 and 19 
show that increasing the window size enhances the system accuracy. This is due to leveraging more information. This, however, increases the latency of the location estimation. Therefore,  an application should balance the latency-accuracy tradeoff based on its requirements.

\subsubsection{Clustering inconsistency threshold ($r$)}

Figures~12 and 20 
show that for small values of $r$, i.e. $r < 0.15$, the system tends to generate one cluster, regardless of the number of entities in the area of interest, underestimating the true number of humans. As $r$ approaches its maximum value, i.e. one, the system generates a lot of clusters, overestimating the actual number of humans. This quantifies the advantage of the clustering module. An optimal value for $r$ occurs around $0.25$.

\subsubsection{Fingerprint density ($n$)}

Figures~13 and 21 
show that increasing the fingerprint density increases accuracy. As small as 15 locations, corresponding to a density of one FP location every $7.6 m^2$, is enough to achieve the best accuracy. Increasing the density beyond this value does not significantly enhance the accuracy.

\subsubsection{Number of streams ($k$)}

Figures~14 and 22 
show that increasing the number of streams increases the system accuracy, especially for a higher number of entities, to a certain limit after which the performance saturates. As few as four streams can achieve less than 1.6 meter overall accuracy for the zone-based difference.
\subsubsection{HMM Order ($o$)}
\label{sec:res_hmm}

Figures~15 and 23 
show that a second order model enhances performance over lower order models. In some cases, a third order model performs worse than a second order mode due to over-training. This justifies the use of a second order HMM.

\subsection{Comparison with Other DF Systems}
\subsubsection{Accuracy}

\begin{table*}[!t]
    \centering
    \caption{Performance summary for the different systems under the two testbeds using the zones-based difference as a metric. Number between parenthesis represent percentage of Spot-One entity advantage. $c$ is the number of candidate locations after the graph-cut phase in Spot and first phase of probabilistic Nuzzer. $c$ is typically $<< n$.}
   \small
 \begin{tabular}{|l||p{1.2cm}|p{1.2cm}|p{1.2cm}||p{1.25cm}|p{1.2cm}|p{1.2cm}||c|} \hline
       & \multicolumn{3}{c||}{Testbed 1}& \multicolumn{3}{c||}{Testbed 2}& \\\cline{2-7}
          & Median &Average& Running &   Median &Average&Running& \\
       \raisebox{3ex}{System} & \raisebox{1.5ex}{error} & \raisebox{1.5ex}{error}& \raisebox{1.5ex}{time}&\raisebox{2ex}{error}  &\raisebox{1.5ex}{error}& \raisebox{1.5ex}{time}&  \raisebox{3ex}{Complexity} \\ \hline \hline
         Spot-One ent.&  1m & 1.43m&1.95ms & 1.1m & 1.25m&1.9ms &\\ \cline{1-7}
         Spot-Multi-ent.&  1.75m (75\%)& 2.15m (50.3\%)& 2.56ms (31.4\%)& 0.85m \mbox{(-22.7\%)}& 1.72m (37.6\%)& 2.4ms (26.3\%)&  \raisebox{1ex}{$O(n.m+c^3)$}\\ \hline \hline
       Prob. Nuzzer \cite{Seifeldin:Nuzzer_Report}& 2.3m (130\%)& 2.66m (86\%)&3.53ms (81.35\%)&  1.5m (36.4\%)&  1.64m (31.2\%)&2.85ms (49.84\%)& $O(n.m+n.c)$\\ \hline
       Det. Nuzzer \cite{Seifeldin:Nuzzer}&  3m (270\%)& 3.54m (147.5\%)&1.78ms \mbox{(-8.4\%)}& 2.7m (145.5\%)&  3.12m (149.6\%)&1.92ms (1.1\%) & $O(n.m)$\\ \hline    \end{tabular}
\label{tab:CDFs_1}
 \end{table*}

Figures~16 and 24 
show the CDF of distance error for the different techniques (\textbf{\emph{note that current state-of-the-art supports only one entity}}). Tables~\ref{tab:CDFs_1} and ~\ref{tab:CDFs_2} summarize the results for the two testbeds.
The results show that \textit{Spot} has the best performance under the two testbeds with an enhancement of at least 36\% in median error over the best state-of-the-art techniques for zones-based difference and at least 15.49\% in average error for locations-based difference. All techniques perform better in Testbed 2 due to the closer separation of training point in Testbed 2.

Figure~17 
also shows that \textit{Spot} can estimate the number of entities in the area of interest with at most one difference error. This can further be enhanced as described in Section~\ref{sec:discuss}.
\subsubsection{Running Time}

Figure~18 
and Table~\ref{tab:CDFs_1} show the running time for the different techniques and Spot components. The results show that the overall Spot operations take less $1.9 ms$ per location estimate for both testbeds. The clustering component consumes the largest time, followed by the min-cut algorithm, and finally calculating the probabilities.

Although Table~\ref{tab:CDFs_1} shows that all algorithms have the same complexity (as $c<<n$), the running time does differ. This is due to the proportionality constants for the small $n$ and $m$ values in our experiment.
 Spot takes slightly higher running time than the deterministic technique (less than 4.75\% on average for both testbeds). However, it significantly outperforms the probabilistic Nuzzer technique, with 65\% enhancement on average in running time. This highlights that Spot significant gain in accuracy and reduction in training overhead comes at a negligible increase in running time.

\begin{table*}[!t]
    \begin{tabular}{p{0.333\textwidth} p{0.333\textwidth} p{0.333\textwidth}}
    \vspace{-3.9cm}\includegraphics[height=1.41in, width=0.32\textwidth]{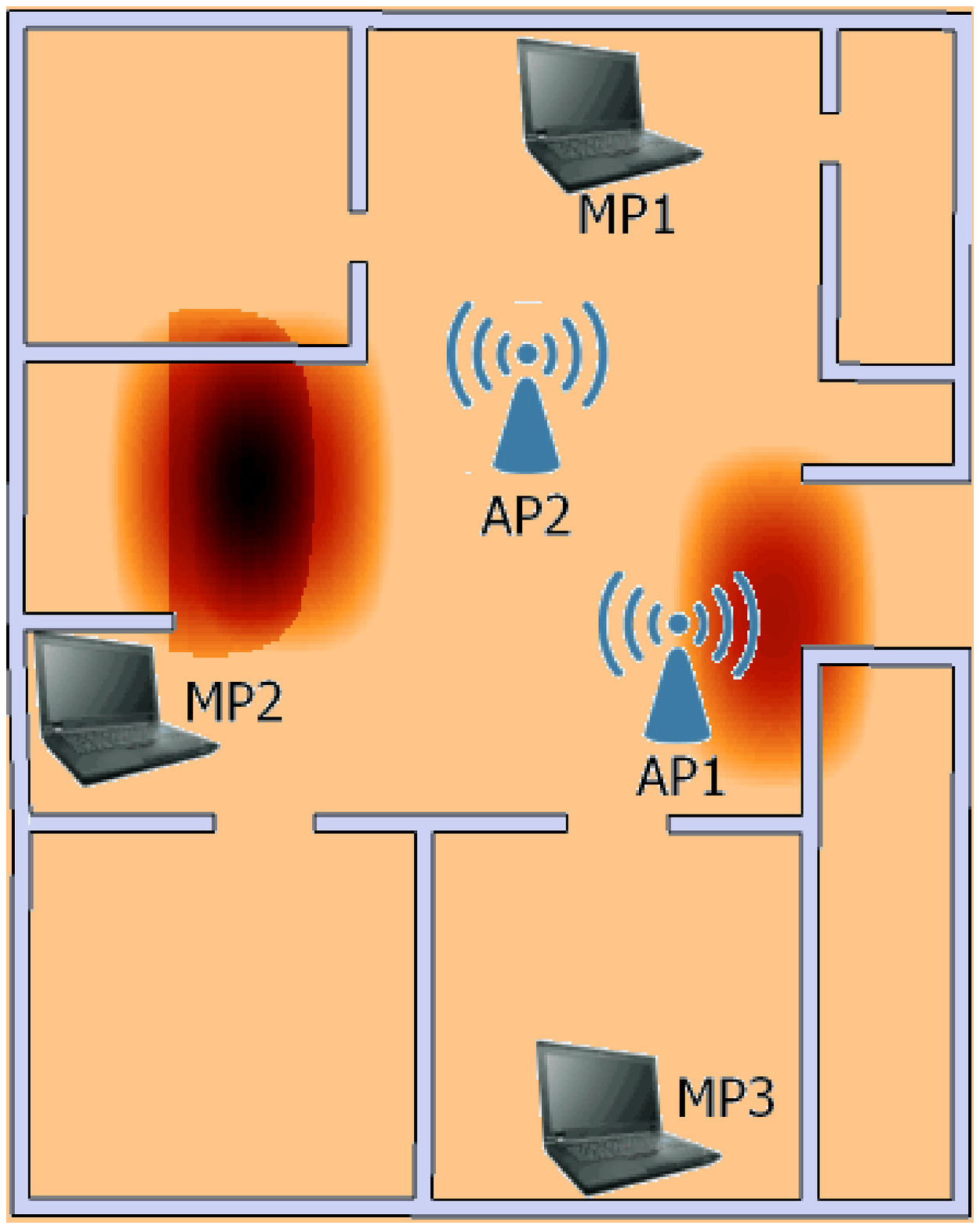}&
    \includegraphics[width=0.33\textwidth]{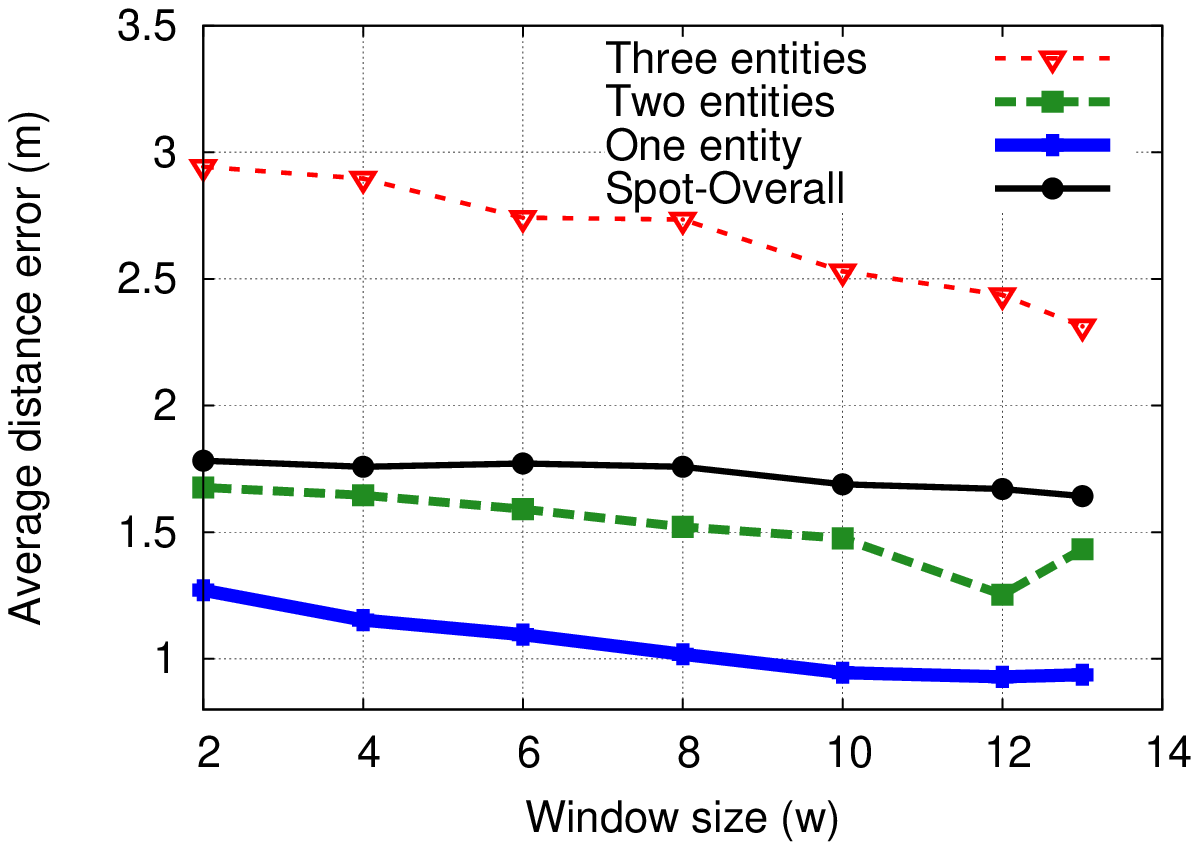}&
    \includegraphics[width=0.33\textwidth]{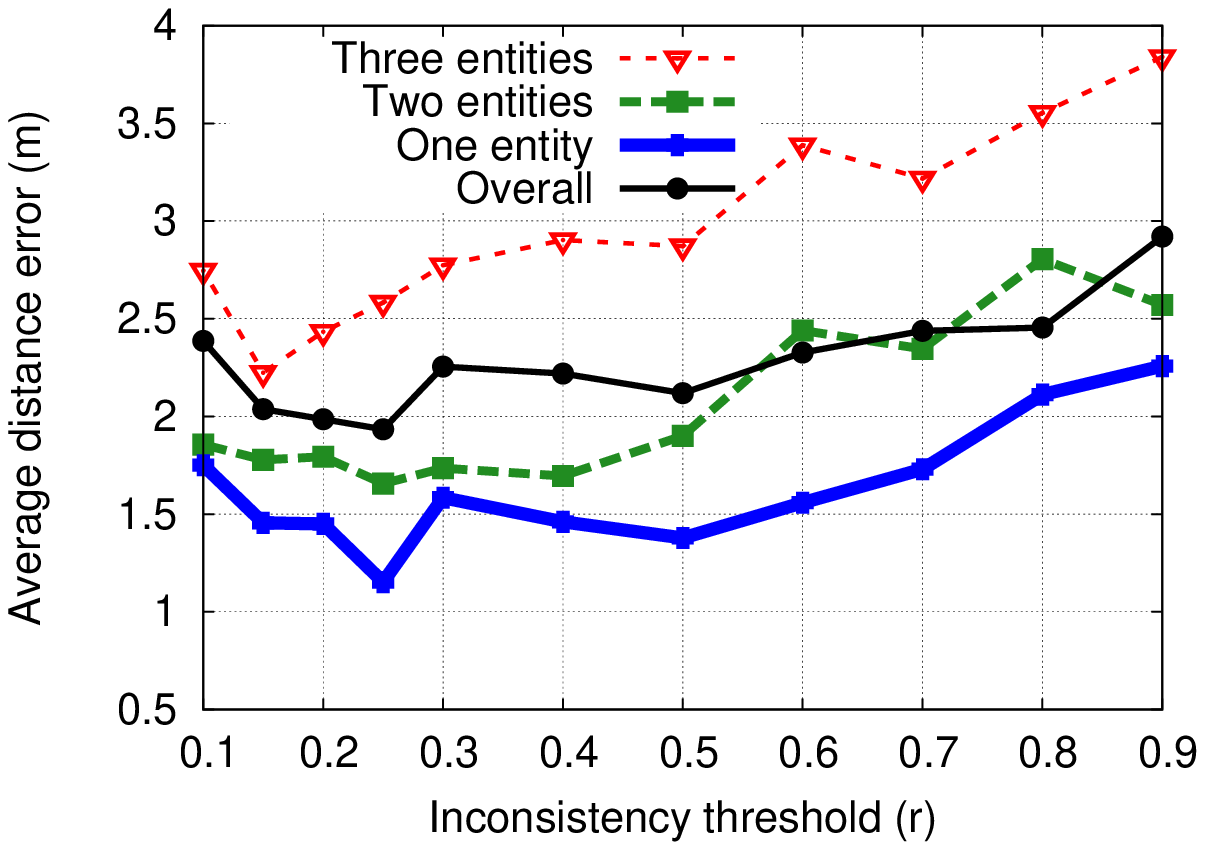}\\
    Fig. 10: A heatmap highlighting the system output. Two close entities are present on the left and another entity is present on the right. & Fig. 11: Effect of changing the window size ($w$) on accuracy. (zones-based difference)& Fig. 12: Effect of changing the clustering inconsistency threshold ($r$) on accuracy.(zones-based difference)\\

    \includegraphics[width=0.33\textwidth]{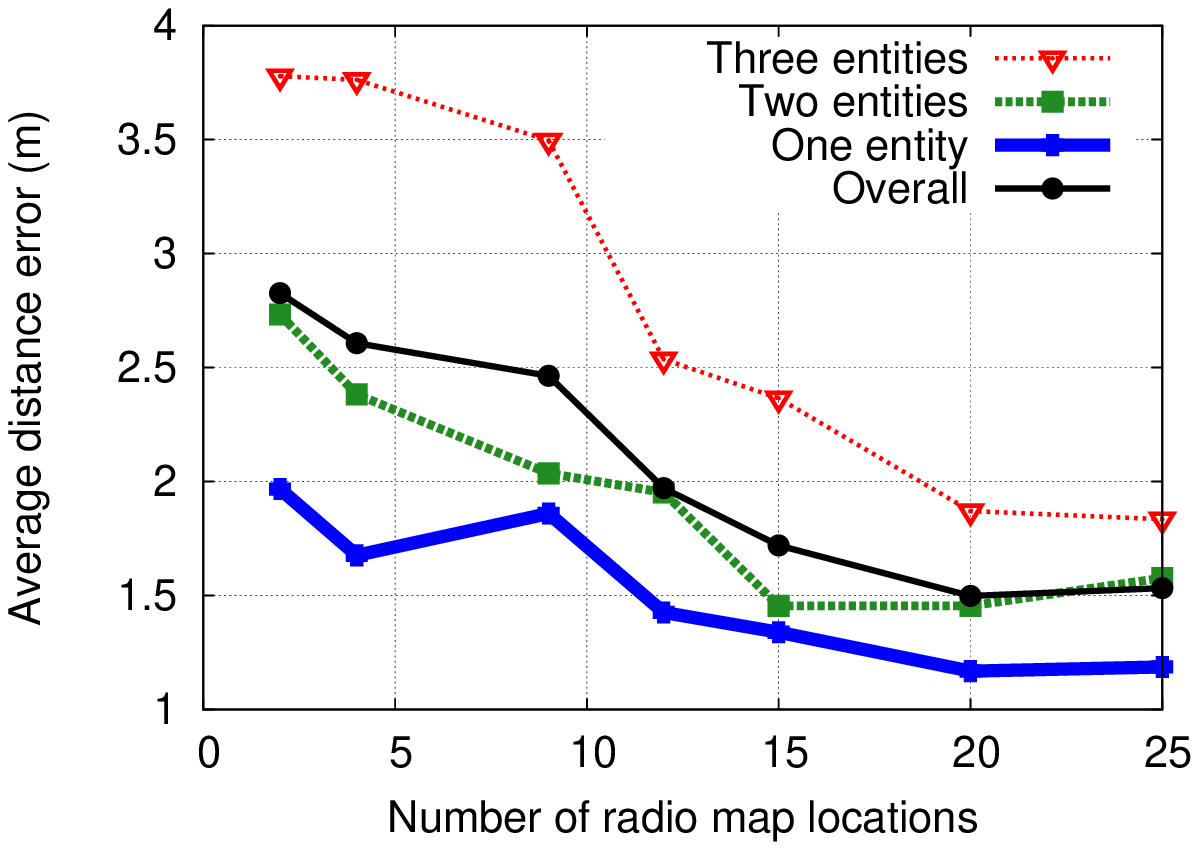} &
    \includegraphics[width=0.33\textwidth]{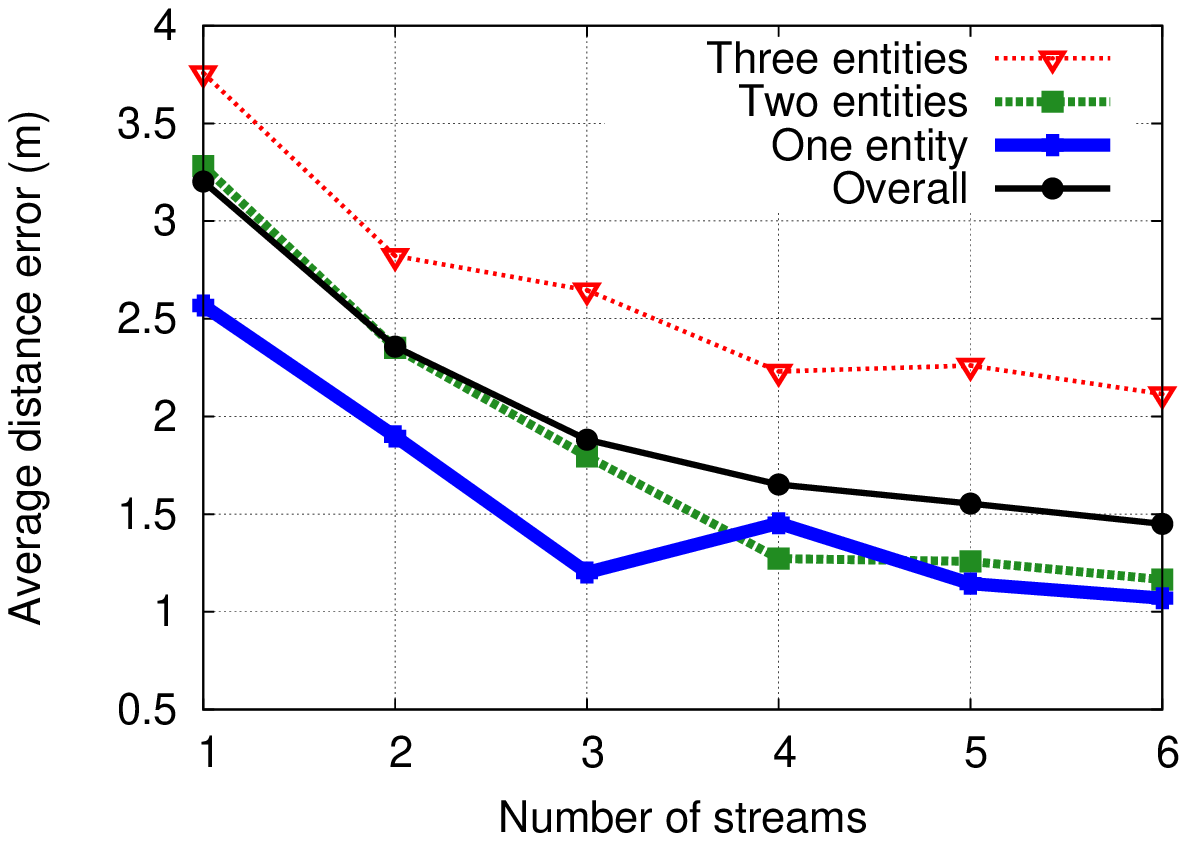} &
    \includegraphics[width=0.33\textwidth]{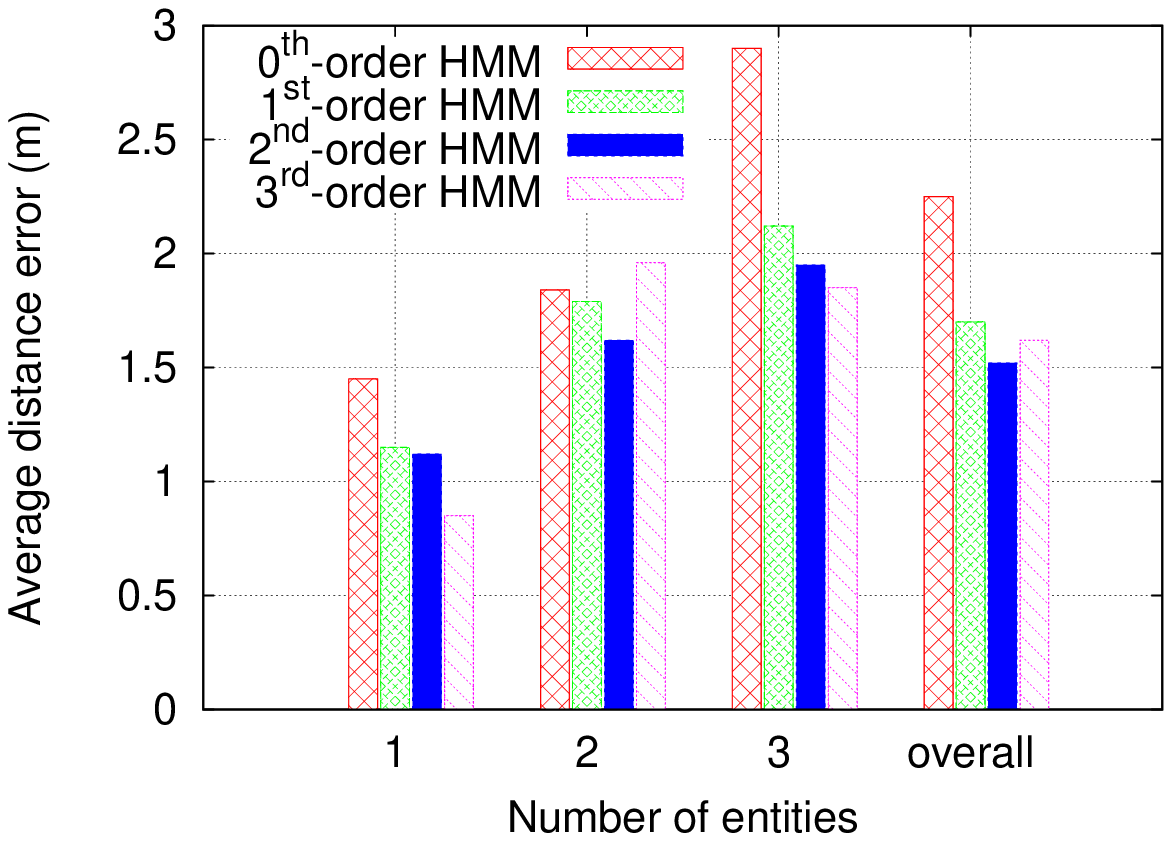}\\

    Fig. 13: Effect of changing the fingerprint density ($n$) on accuracy. (zones-based difference)& Fig. 14: Effect of changing the number of streams ($k$) on accuracy. (zones-based difference)& Fig. 15: Effect of changing the HMM order ($o$) on accuracy. (zones-based difference)\\

    \includegraphics[width=0.33\textwidth]{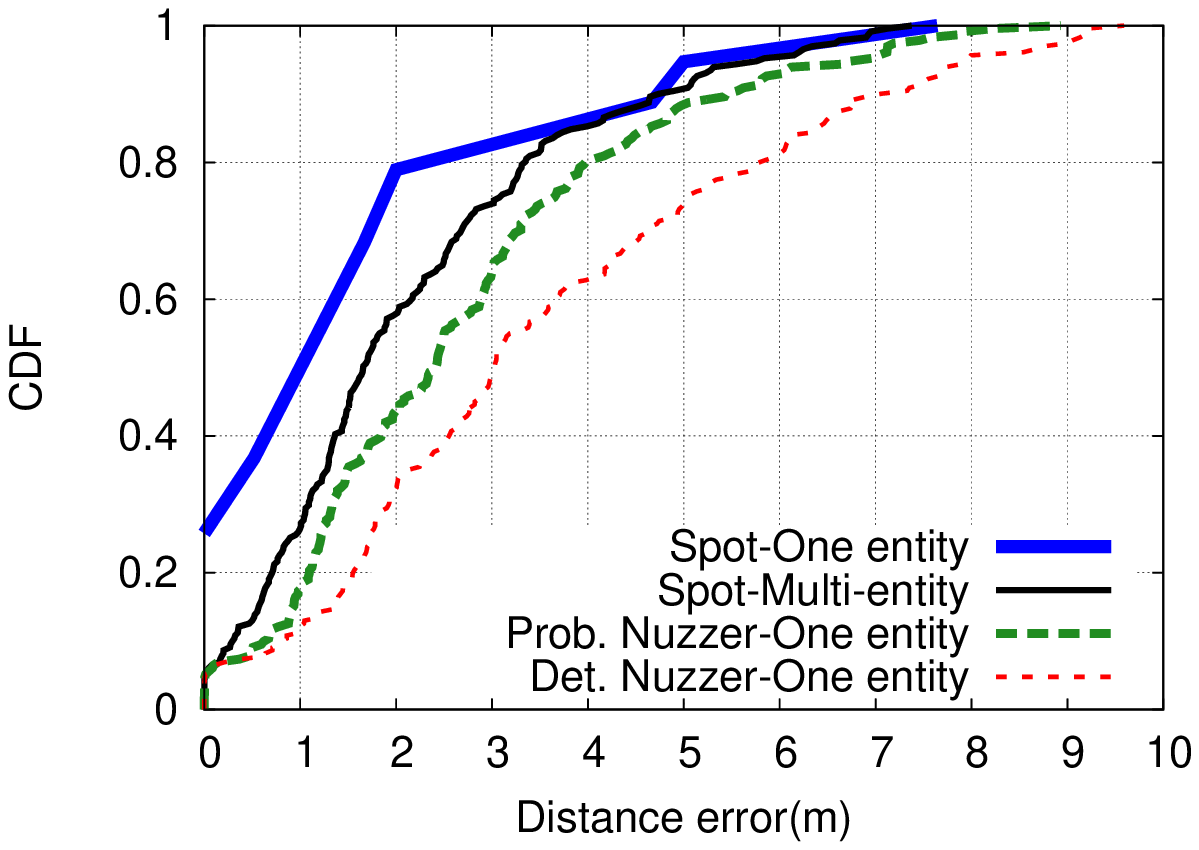}&
    \includegraphics[width=0.33\textwidth]{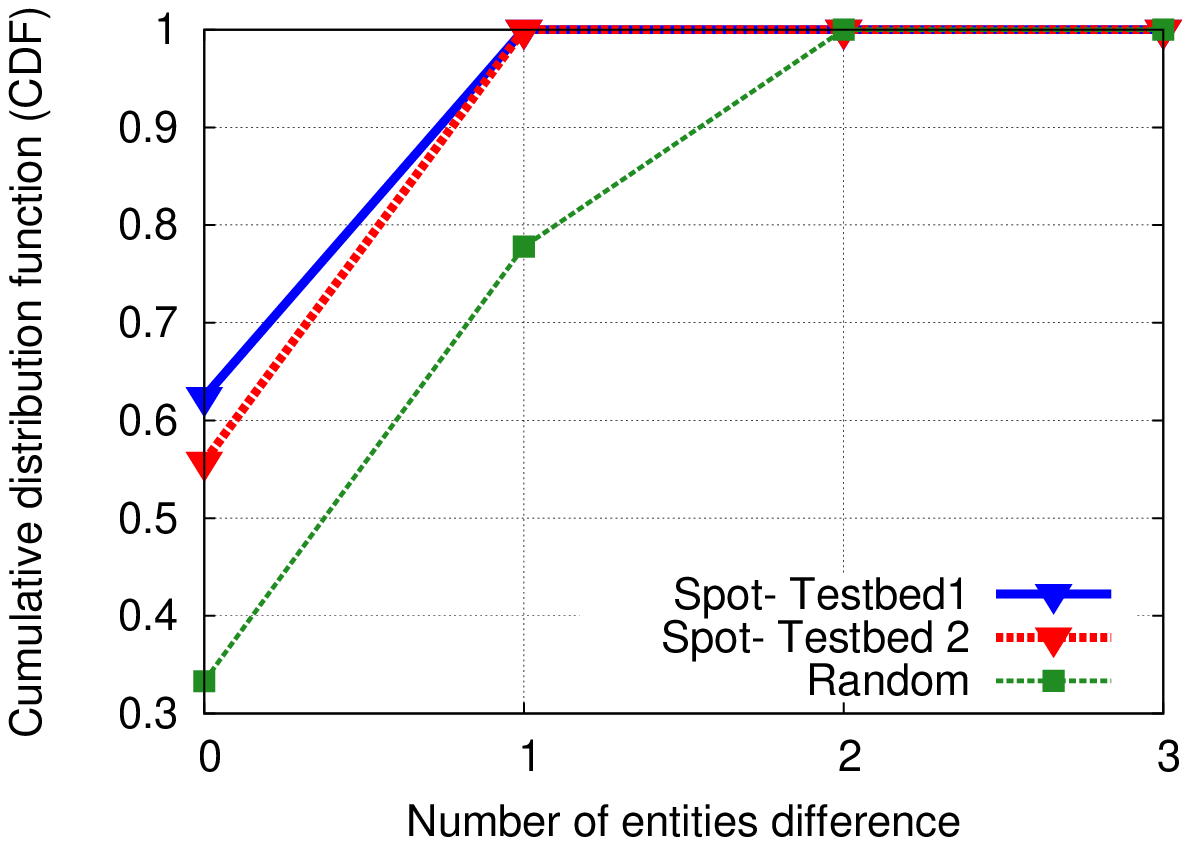}&
    \includegraphics[width=0.33\textwidth]{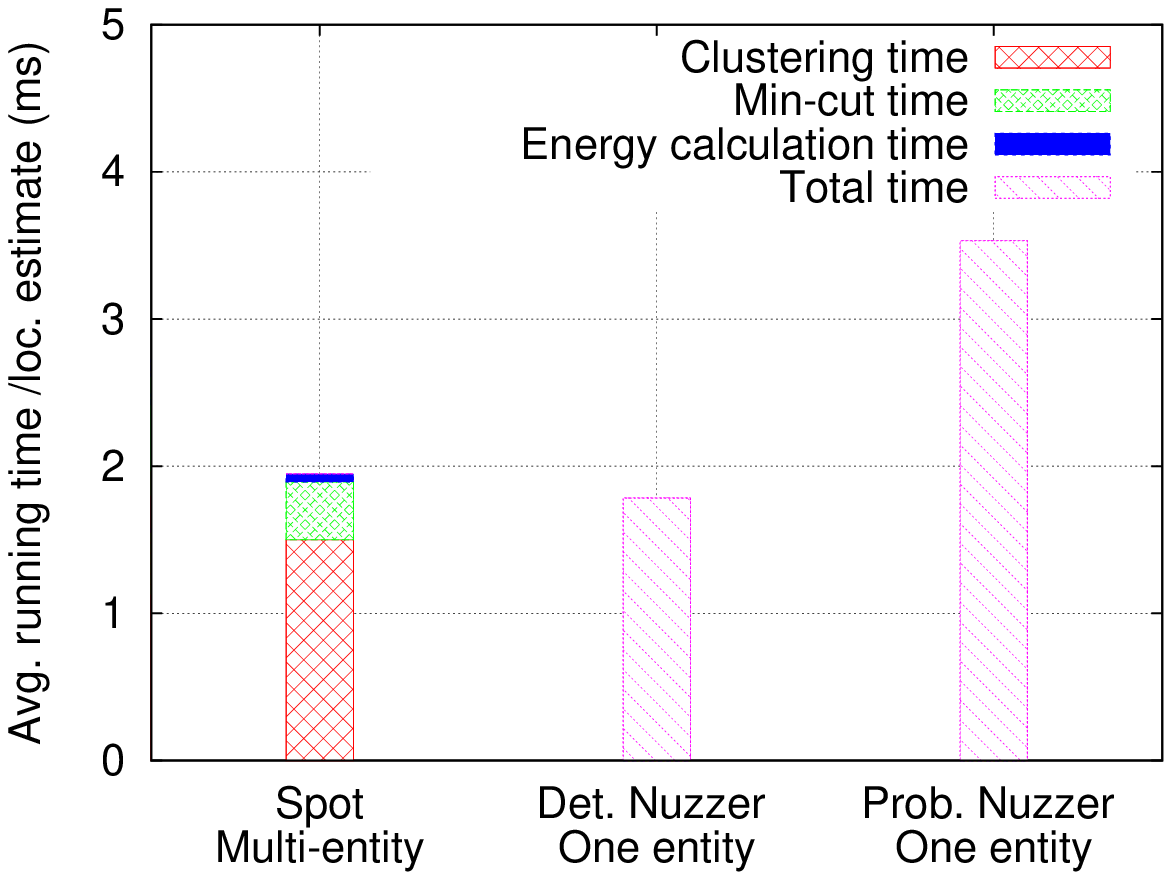}\\
    Fig. 16: CDF of distance error for Testbed 1. (zones-based difference)& Fig. 17: CDF of num. of entities estimation error. A random estimator is a baseline that presents a lower bound on performance. & Fig. 18: Running time for the different components of Spot and a comparison with other systems running time.\\

    \end{tabular}
\end{table*}

\begin{table*}[!t]
    \centering
    \caption{Performance summary for the different systems under the two testbeds using locations-based difference as a metric. Number between parenthesis represent percentage of Spot-One entity advantage. $c$ is the number of candidate locations after the graph-cut phase in Spot and first phase of probabilistic Nuzzer. $c$ is typically $<< n$. }
   \small
 \begin{tabular}{|l||p{1.2cm}|p{1.2cm}|p{1.2cm}||p{1.25cm}|p{1.2cm}|p{1.2cm}||c|} \hline
       & \multicolumn{3}{c||}{Testbed 1}& \multicolumn{3}{c||}{Testbed 2}& \\\cline{2-7}
          & Median &Average& Running &   Median &Average&Running& \\
       \raisebox{3ex}{System} & \raisebox{1.5ex}{error} & \raisebox{1.5ex}{error}& \raisebox{1.5ex}{time}&\raisebox{2ex}{error}  &\raisebox{1.5ex}{error}& \raisebox{1.5ex}{time}&  \raisebox{3ex}{Complexity} \\ \hline \hline
         Spot-One ent.& 1.28  m & 1.52 m&1.95ms &  1.36m & 1.42m&1.9ms &\\ \cline{1-7}
         Spot-Multi-ent.& 2.1m (64\%)& 2.26m (48.68\%)& 2.56ms (31.4\%)& 1.22m \mbox{(-10.2\%)}& 1.94m (36.61\%)& 2.4ms (26.3\%)&  \raisebox{1ex}{$O(n.m+c^3)$}\\ \hline \hline
       Prob. Nuzzer \cite{Seifeldin:Nuzzer_Report}& 2.3m (79.68\%)& 2.66m (75\%)&3.53ms (81.35\%)&  1.5m (10.29\%)&  1.64m (15.49\%)&2.85ms (49.84\%)& $O(n.m+n.c)$\\ \hline
       Det. Nuzzer \cite{Seifeldin:Nuzzer}&  3m (134.37\%)& 3.54m (132.8\%)&1.78ms \mbox{(-8.4\%)}& 2.7m (98.52\%)&  3.12m (119.71\%)&1.92ms (1.1\%) & $O(n.m)$\\ \hline    \end{tabular}
\label{tab:CDFs_2}
 \end{table*}

\begin{table*}[!t]
    \begin{tabular}{p{0.333\textwidth} p{0.333\textwidth} p{0.333\textwidth}}
    \vspace{-3.9cm} \includegraphics[width=0.33\textwidth]{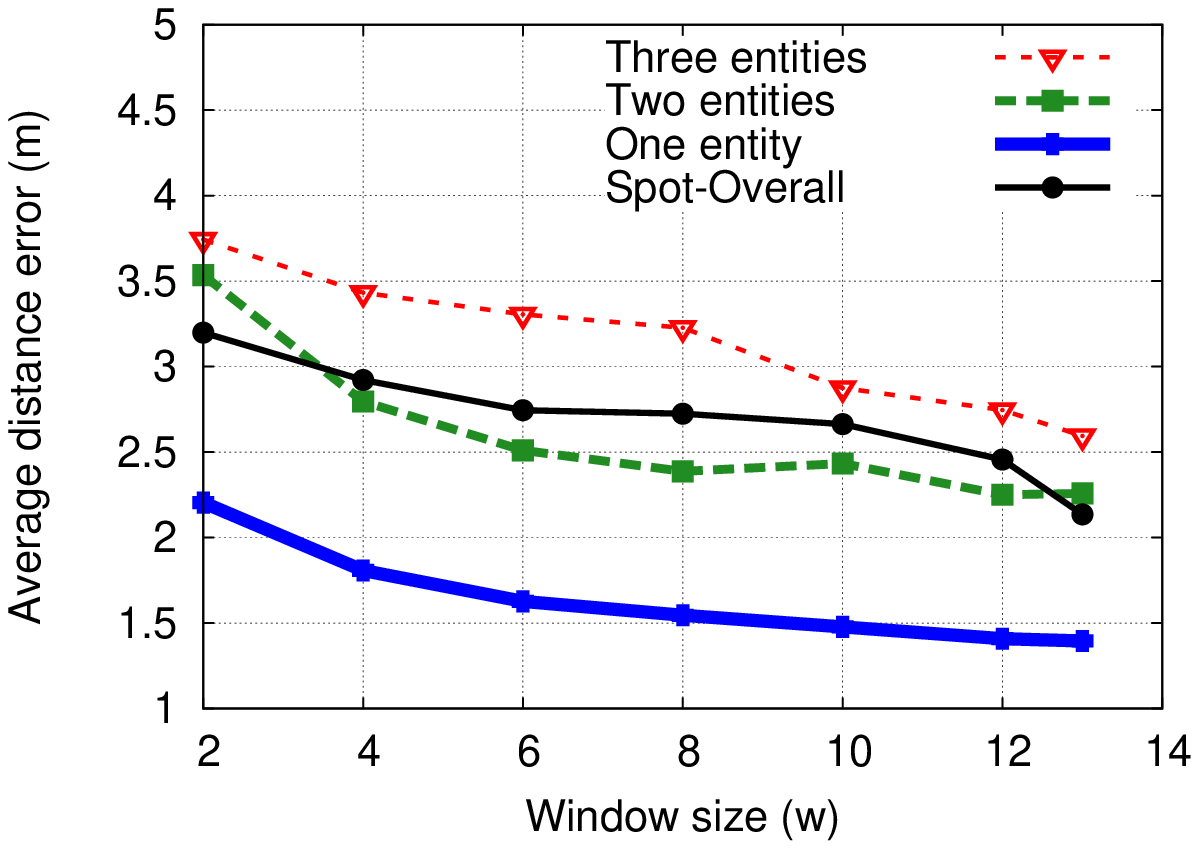}&
    \includegraphics[width=0.33\textwidth]{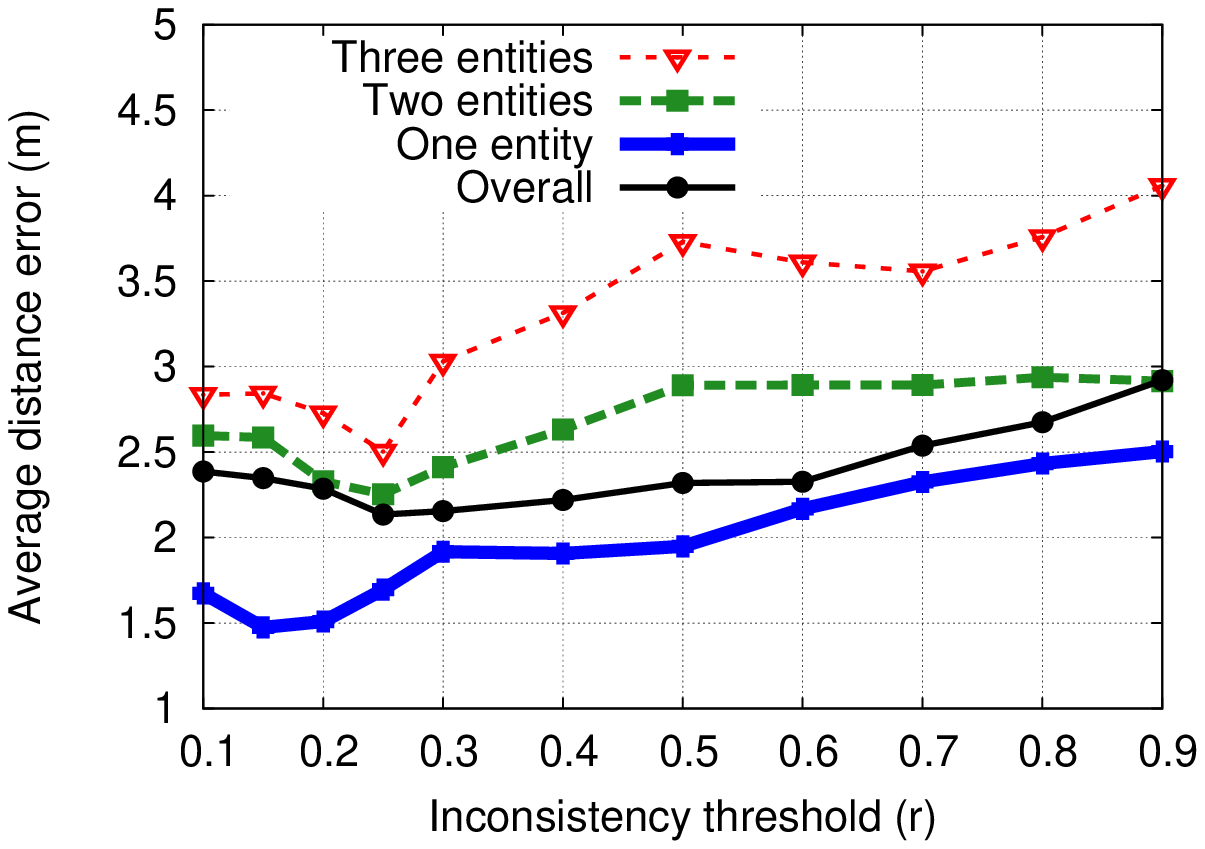} &
    \includegraphics[width=0.33\textwidth]{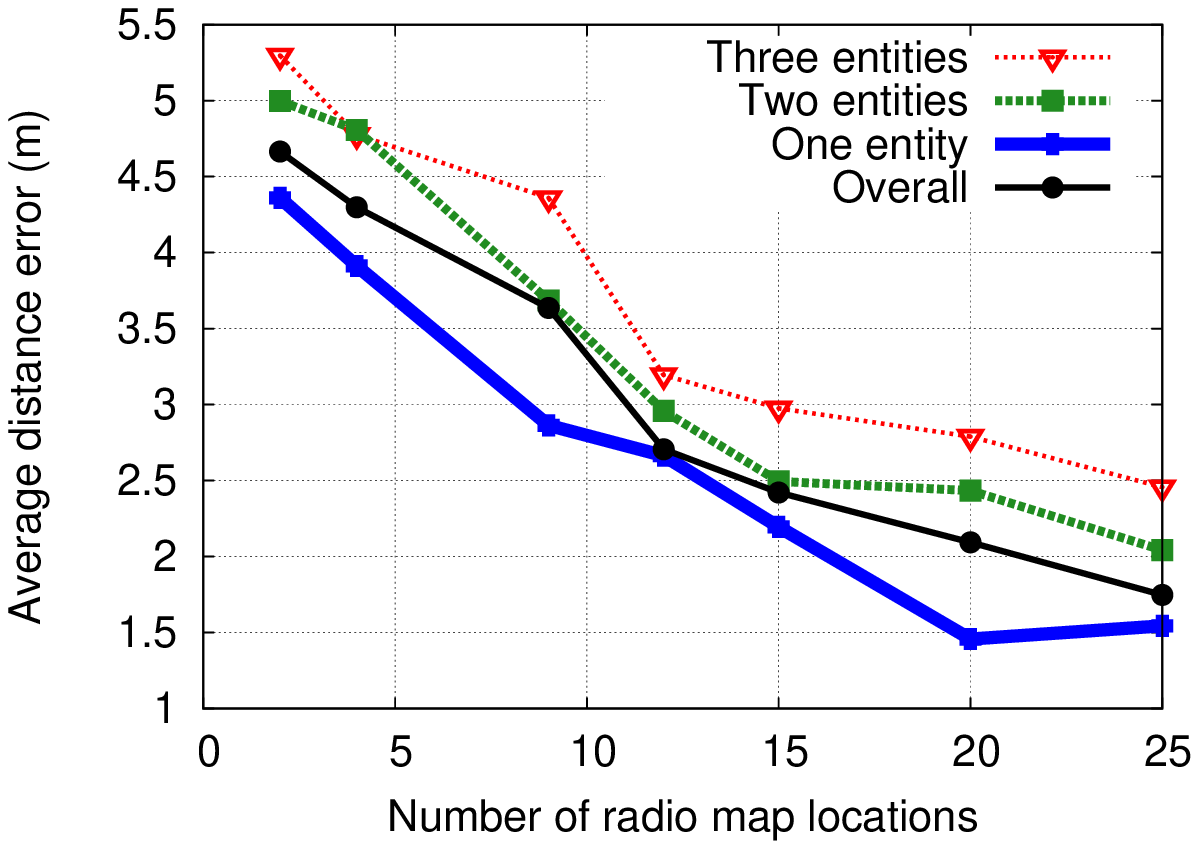} \\
    Fig. 19: Effect of changing the window size ($w$) on accuracy. (locations-based difference)& Fig. 20: Effect of changing the clustering inconsistency threshold ($r$) on accuracy. (locations-based difference)& Fig. 21: Effect of changing the fingerprint density ($n$) on accuracy. (locations-based difference)\\

    \includegraphics[width=0.33\textwidth]{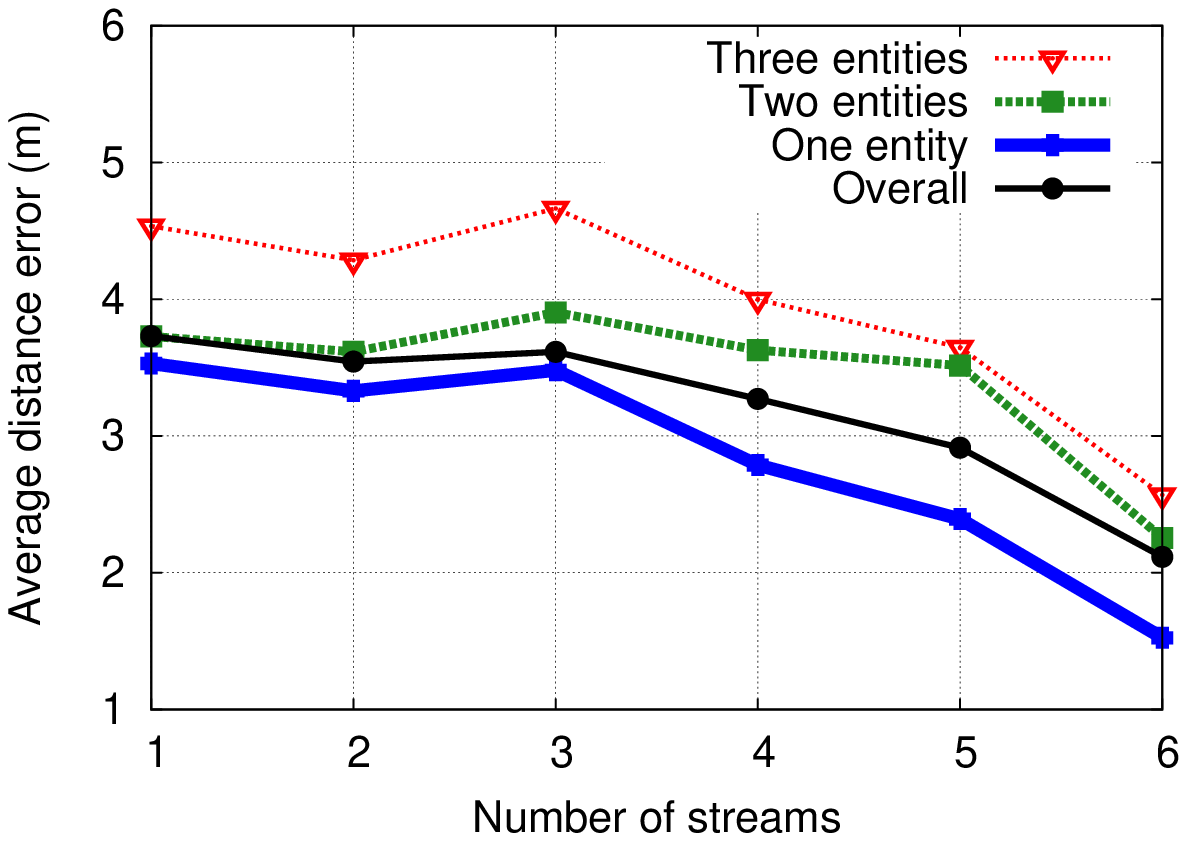} &
    \includegraphics[width=0.33\textwidth]{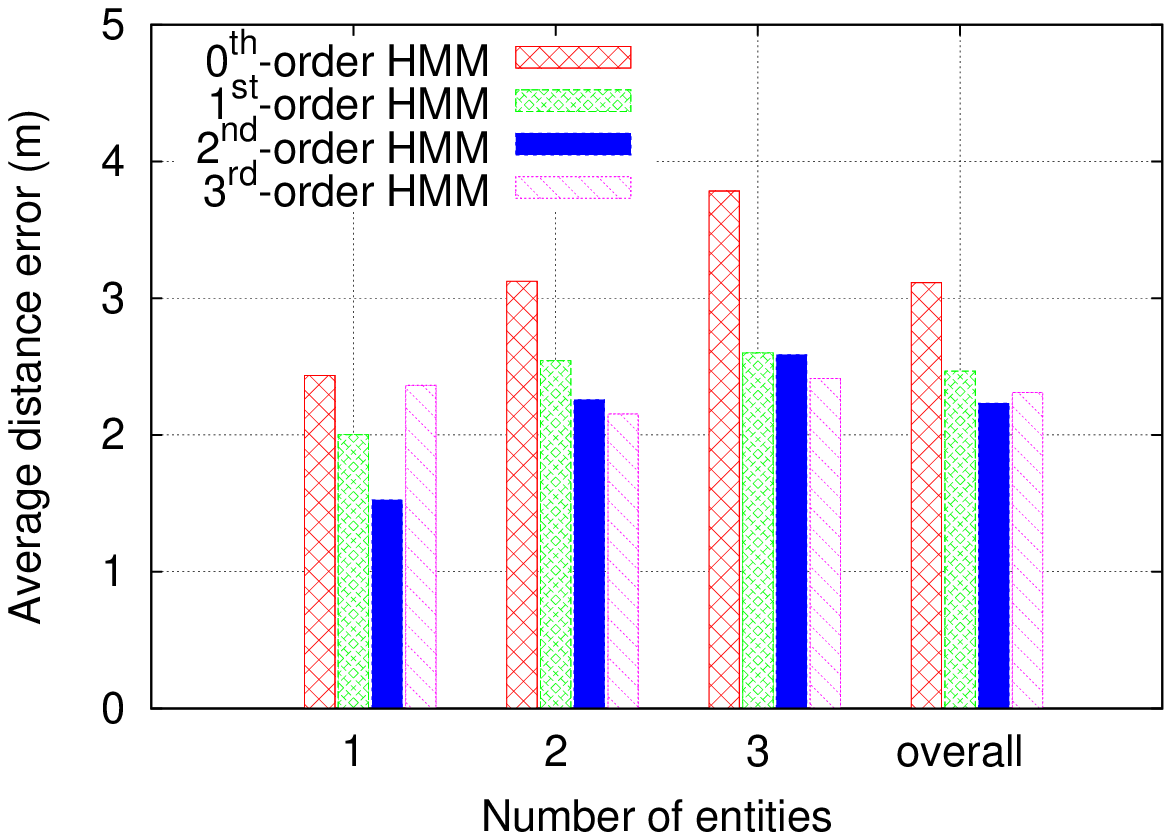} &
    \includegraphics[width=0.33\textwidth]{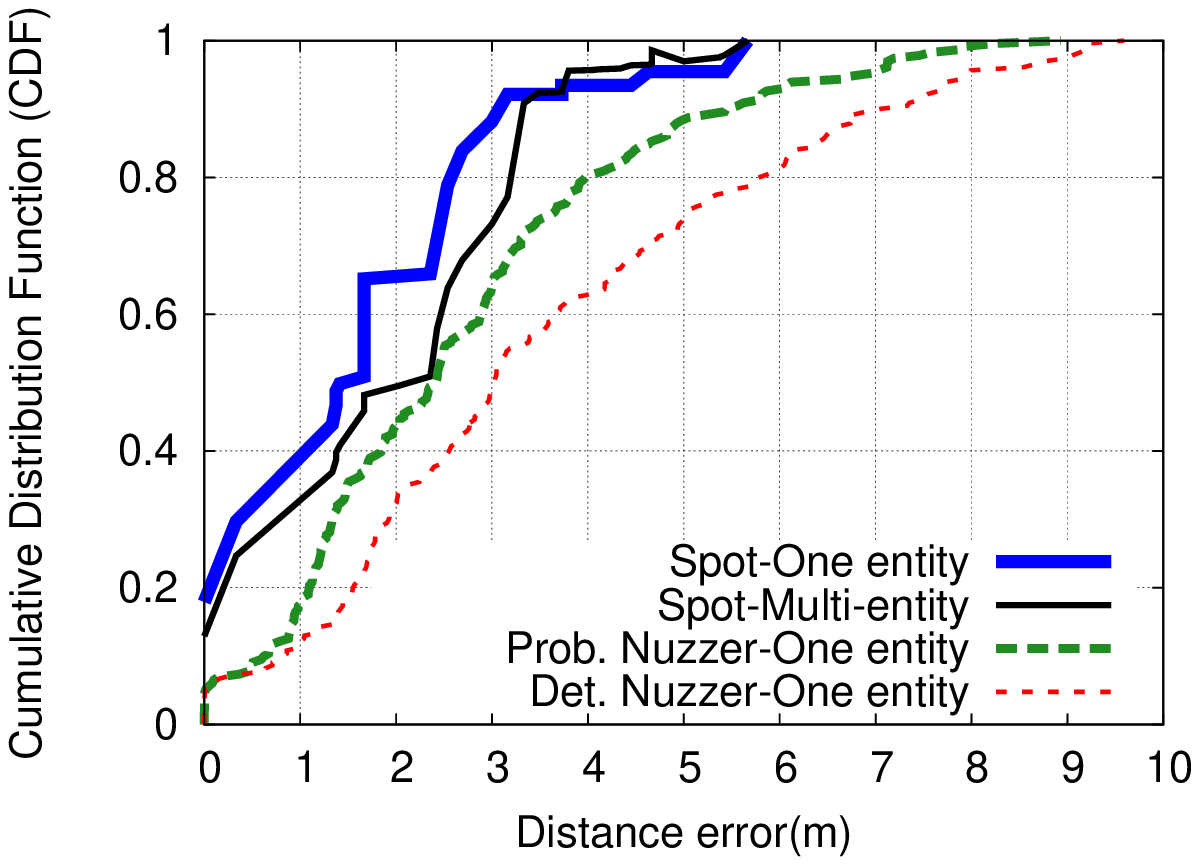}\\

    Fig. 22: Effect of changing the number of streams ($k$) on accuracy. (locations-based difference)& Fig. 23: Effect of changing the HMM order ($o$) on accuracy. (locations-based difference)& Fig. 24: CDF of distance error for Testbed 1. (locations-based difference)\\

    \end{tabular}
\end{table*}

\section{Discussion}
\label{sec:discuss}
In this section, we discuss different aspects of \textit{Spot}.

\subsection{Path Training}
Using the proposed framework, we could reduce the training complexity from $O(2^n)$ to $O(n)$. This is a significant reduction in the calibration overhead which turns the multi-entity tracking problem to a feasible problem. However, there is still some effort in calibrating the area of interest as the user has to stand at each location for a certain time. One possibility to reduce this overhead is to use path-based training, where a user continuously moves between two points and samples are collected along the path. This continuous calibration reduces the overhead, but provides less samples. Multiple passes around the area of interest can be used to increase the number of available samples along with density interpolation between adjacent locations. Further experiments need to be performed though to asses the tradeoffs of this technique.

\subsection{Identification}
Although we can track multiple entities in the area of interest, identifying these entities remains an open problem. This identification includes knowing the entities' physical identity (e.g. its name) or virtual identity, i.e. associating a unique ID to the detected entity. This entity labeling problem is well known in other fields, such as computer vision \cite{BerclazFTF11}. The entities movement history and trajectories can be used to detect these virtual identities.\\

\subsection{Number of Entities History Model}
\textit{Spot} can correctly estimate the number of entities with high accuracy. This can be further enhanced based on adding constraints for the temporal smoothness of the number of entities. In other words, outliers in estimating the number of entities can be detected based on the history of the detected number of entities. This can handle cases such as clusters merging and splitting.

\begin{table*}
    \centering
    \caption{Comparison of different RF-based DF localization systems.}
    \begin{tabular}{|l|l|l|l||l|} \hline
      &{\bf MIMO Radar-} &{\bf Radio Tomographic} &{\bf Nuzzer} &{\bf Spot} \\
      &{\bf based Systems} & {\bf Imaging (RTI)} &{\bf System} &{\bf System}\\ \hline \hline
        Special hardware required &  Yes   &  Yes   & No  & No \\ \hline
        Number of special nodes&  Few   & Many   & None  & None \\ \hline
        Number of streams & N/A (echo based)  &  Large (756)   & Small (6)  & Small (6)\\ \hline
        Coverage area&  Limited (high freq.) &   Limited & Yes  & Yes\\
                                              \hline
        Computational Complexity &  Low   & High  & Moderate  & Low \\ \hline
        Accuracy &  Very High   &  High   & Moderate  & High \\ \hline
        Muli-path effect& Limited & Yes & Limited (F.print)& Limited (F.print)\\\hline
        Multi-entity tracking&  Yes   &  Yes   & No  & Yes\\ \hline
        Multi-entity overhead&  Low   &  Low   & Intractable  & Moderate (F.print)\\ \hline
    \end{tabular}
    \label{tab:related}
\end{table*}

\section{Related Work}
\label{sec:related}
Device-free tracking systems have been introduced over the year including: radar-based, camera-based, sensors-based, and WLAN-based systems. Table~\ref{tab:related} shows how \textit{Spot} compares to the different systems.

In the radar-based systems, pulses of radio waves are transmitted into the area of interest and based on measuring the received reflections, objects could be tracked. Several technologies have been presented in this class including ultra-wideband (UWB) systems \cite{Yang:UWB}, doppler radar \cite{Lin:Doppler}, and MIMO radar systems \cite{Haimovich:MIMO}.

Camera-based tracking systems are based on analyzing a set of captured images to estimate the current locations of objects of interest \cite{Moeslund:Camera,Krumm:Camera}. The analysis consists of two main processes: background subtraction and temporal correspondence.

Sensor-based systems use especially installed sensor nodes to cover the area of interest. For example, \cite{Patwari:RTI} applies radio tomographic techniques to the readings of a dense array of sensors to obtain accurate DF tracking.

All the technologies above share the requirement of installing special hardware to be able to perform DF tracking, which reduces their scalability in terms of cost and coverage area. In contrast, WLAN DF tracking aims at exploiting the already installed WLAN.
The DF localization in WLANs was first introduced in \cite{Youssef:DFPchallenges} along with feasibility experiments in a controlled environment. Several papers followed the initial vision to provide different techniques for detection and tracking \cite{Moussa:smart,Yang:Performin,Kosba:RASID,Seifeldin:Nuzzer}. However, all these techniques focus on the problem of a \textbf{\emph{single entity}}. Tracking multiple entities, to-date, has been considered an intractable problem due to the exponential increase in the number of training combination required.

\textit{Spot}, on the contrary, is designed to provide accurate and efficient, i.e. linear training complexity, multi-entity DF localization for WLAN environments.

\section{Conclusion}
\label{sec:conclude}
We presented the design, analysis, and implementation of Spot: a system for accurate and efficient multi-entity device-free WLAN localization. \textit{Spot} leverages probabilistic techniques to provide a smooth and consistent environment image. It uses a cross-calibration technique and an energy minimization framework to reduce the calibration overheard to linear in the number of locations, which turns the DF multi-entity tracking to a tractable problem. We showed that the selected energy minimization terms lead to an efficient solution by mapping the energy function to a binary graph-cut problem. We further showed how to perform clustering on the generated environment map to remove outliers and enhance accuracy.

Implementation on standard WiFi hardware in two different testbeds show that \textit{Spot} can achieve 1.1m median distance multi-entity tracking error, which is better than the stat-of-art techniques by at least 36\% in both testbeds for zone-based differences and 15.49\% in average error for the locations-based difference. In addition, it can estimate the  number of entities correctly to within one entity difference 100\% of the time. This highlights the promise of \textit{Spot} for a wide range of multi-entity DF tracking applications.

\normalsize
\bibliographystyle{abbrv}
\bibliography{DFPTracking}
\appendix
\section{Proof of Theorem 2}
The proof is based on showing that the optimization problem solved by the min-cut on the constructed graph is equivalent to the optimization problem in Equation~\ref{eq:e_i}.

\begin{proof}
  For any node $x$ in the constructed graph (Figure~\ref{fig:graph}), if this node is assigned the label $S$, i.e. its value is $\alpha_x^t= 1$, after running the min-cut algorithm, then part of its corresponding contribution in the optimal cost is $P(s^t|\alpha_x^t=1)+ P(\alpha_x^t=1| \alpha_x^{t-1}, \alpha_x^{t-2})$. Similarly, if this node is assigned the label $T$, i.e. its value is $\alpha_x^t= 0$, after running the min-cut algorithm, then part of its corresponding contribution in the optimal cost is $P(s^t|\alpha_x^t=0)+ P(\alpha_x^t=0| \alpha_x^{t-1}, \alpha_x^{t-2})$. Both are equivalent to the unary terms in Equation~\ref{eq:e_i}.

  Now consider any two nodes $x$ and $y$ in the constructed graph (Figure~\ref{fig:graph}), if these two nodes have the same label, no extra terms are contributed to the optimal cost in the minimal cut. However, if $x$ is assigned to $S$ and $y$ is assigned to $T$ or vice versa, an extra term will be added to the optimal cost in the minimal cut, which is equal to $\frac{1 + e^{-\left\|P(s^t|\alpha_{x}^{t})-P(s^t|\alpha_{y}^{t})\right\|^2}}{2}$. This corresponds to the binary term in Equation~\ref{eq:e_i}.

Therefore the cost of the minimal cut in the constructed graph is equivalent to the minimum of the summation of both the unary and binary terms. Therefore the binary graph-cut solution on the constructed graph is a solution to the corresponding energy minimization problem in Equation~\ref{eq:eng_min}.
\end{proof}

\end{document}